\documentclass[11pt]{article}

\usepackage{lmodern}
\usepackage[T1]{fontenc}
\usepackage[latin9]{inputenc}
\usepackage{geometry}
\usepackage{color}
\usepackage{float}
\usepackage{amsmath}
\usepackage{amsthm}
\usepackage{amssymb}
\usepackage[unicode=true,
 bookmarks=false,
 breaklinks=false,pdfborder={0 0 1},backref=section,colorlinks=false]
 {hyperref}
\usepackage{algorithmic}
\usepackage{subcaption} 
\usepackage{graphicx}
\usepackage{epstopdf}
\usepackage{hyperref}
\usepackage{booktabs}
\usepackage{amsmath}
\usepackage{dsfont}
\usepackage{bm}
\usepackage{stfloats}
\usepackage{xcolor,colortbl}
\usepackage{epstopdf}

\newcommand{\unconstrained}{{\sc Unconstrained}}
\newcommand{\constrained}{{\sc $k$-Constrained}}
\newcommand{\matroidconstrained}{{\sc Matroid-Constrained}}
\newcommand{\unconstrainedonline}{{\sc Online-Unconstrained}}
\newcommand{\constrainedonline}{{\sc Streaming-$k$-Constrained}}

\newcommand{\cov}{{{\ensuremath{f}}}}
\newcommand{\cost}{{{\ensuremath{c}}}}
\newcommand{\opt}{{{\ensuremath{\mathrm{OPT}}}}}

\newcommand{\costscaledgreedy}{{\texttt{CSG}}}
\newcommand{\onlinecostscaledgreedy}{{\texttt{Online-CSG}}}
\newcommand{\onlinekcostscaledgreedy}{{\texttt{Streaming-CSG}}}
\newcommand{\costscaledlazygreedy}{{\texttt{CSLG}}}
\newcommand{\matroidcostscaledgreedy}{{\texttt{MCSG}}}

\newcommand{\distortedgreedy}{{\texttt{DistortedGreedy}}}
\newcommand{\stochasticdistortedgreedy}{{\texttt{StochasticDistortedGreedy}}}
\newcommand{\unconstraineddistortedgreedy}{{\texttt{UnconstrainedDistortedGreedy}}}
\newcommand{\topkexperts}{{\texttt{TopK}}}

\newcommand{\greedy}{{\texttt{Greedy}}}

\newcommand{\const}{{\ensuremath{s}}}
\newcommand{\GuruDataset}{{\textit{Guru}}}

\newcommand{\InfluenceDataset}{{\textit{NetHEPT}}}
\newcommand{\YelpDataset}{{\textit{Yelp}}}
\newcommand{\MovielensDataset}{{\textit{Movielens}}}

\newcommand{\etal}{{\emph{et al.}}}

\providecommand{\definitionname}{Definition}
\providecommand{\lemmaname}{Lemma}
\providecommand{\theoremname}{Theorem}
\newtheorem{problem}{Problem}

\newtheorem{thm}{\protect\theoremname}[section]
\newtheorem{lem}[thm]{\protect\lemmaname}
\floatstyle{ruled}
\newfloat{algorithm}{tbp}{loa}
\providecommand{\algorithmname}{Algorithm}
\floatname{algorithm}{\protect\algorithmname}

\newcommand{\spara}[1]{\smallskip\noindent{\bf{#1}}}

\newcommand*{\belowrulesepcolor}[1]{%
	\noalign{%
		\kern-\belowrulesep
		\begingroup
		\color{#1}%
		\hrule height\belowrulesep
		\endgroup
	}%
}
\newcommand*{\aboverulesepcolor}[1]{%
	\noalign{%
		\begingroup
		\color{#1}%
		\hrule height\aboverulesep
		\endgroup
		\kern-\aboverulesep
	}%
}
\newcommand{\squishlist}{\begin{list}{$\bullet$}
  { \setlength{\itemsep}{0pt}
     \setlength{\parsep}{3pt}
     \setlength{\topsep}{3pt}
     \setlength{\partopsep}{0pt}
     \setlength{\leftmargin}{1.5em}
     \setlength{\labelwidth}{1em}
     \setlength{\labelsep}{0.5em} } }
\newcommand{\squishend}{
  \end{list}  }

\allowdisplaybreaks

\begin{document}

\title{An Efficient Framework for Balancing Submodularity and Cost \thanks{Department of Computer Science, Boston University. \{smnikol,aene,evimaria\}@bu.edu}}
\author{Sofia Maria Nikolakaki \and Alina Ene \and Evimaria Terzi}
\date{}

\maketitle

\begin{abstract}
In the classical \emph{selection problem}, the input consists of a collection of elements and the goal is to
pick a subset of elements from the collection such that some objective function $\cov$ is maximized.
This problem has been studied extensively in the data-mining community 
and it has multiple applications including influence maximization in social networks, team formation and 
recommender systems.
A particularly popular formulation that captures the needs of many such applications is one where
the objective function $\cov$ is a monotone and non-negative submodular function.  In these cases, 
the corresponding computational problem can be solved using a simple greedy
$(1-\frac{1}{e})$-approximation algorithm. 

In this paper, we consider a generalization of the above formulation
where the goal is to optimize a function that maximizes the 
submodular function $\cov$ minus a linear cost function $\cost$. 
This formulation appears as a more natural one, particularly when one needs to 
strike a balance between the value of the objective function and the cost being paid in order to pick 
the selected elements.   
We address variants of this problem both in an offline setting, where the 
collection is known apriori, 
as well as in online settings, where the elements of the collection arrive in an online fashion.   
We demonstrate that by using simple variants of the standard greedy algorithm (used for submodular optimization) we can design algorithms that have provable approximation guarantees, are extremely efficient and work very well in practice.   
\end{abstract}

\section{Introduction}\label{sec:introduction}
The \emph{element selection} problem is central in the data-mining community and it essentially seeks to pick
a set of 
elements from a collection so that some objective is optimized. The applications of such a general formulation abound.
The most relevant to this work are those related to \emph{influence maximization} in social networks~\cite{kempe2003maximizing, leskovec2007cost,  li2018influence}, \emph{team formation}~\cite{anagnostopoulos2012online,  kargar2013finding, lappas2009finding} and \emph{recommender systems}~\cite{borodin2012max, dasgupta2013summarization,kazemi2020regularized}.
For example, in influence maximization the goal is to pick a subset of the nodes of the network so that once they adopt an  item (e.g., idea or product) the spread of this item in the network is maximized.
Similarly, in team formation, given a collection of experts the goal is to pick a subset of them such that they cover the skills
required for a task and also optimize a social objective.
Finally, in recommender systems the goal is to pick a subset of items from a collection (e.g., restaurants or movies) such that the selected items best summarize the collection or best match the users' interests.

In many of the above examples the problem is formulated as a submodular-optimization problem where the goal is 
to pick a subset of $k$ elements $Q$ from a collection $V$ such that $\cov(Q)$ is maximized, where $\cov(Q)$ is a monotone and non-negative submodular function. Subsequently, an easy-to-implement and practical  
greedy algorithm is used to solve such problems
and provide a solution with approximation guarantee $(1-\frac{1}{e})$.
 
\spara{Conceptual contributions:} In this paper, we consider a generalization of this framework, where the goal is again to pick a set of elements $Q$ from 
an input collection $V$. However, our goal is to not only maximize the function $\cov(Q)$, but also to strike a balance between 
the benefits of choosing $Q$, as quantified by $\cov(Q)$, and the cost of picking $Q$, denoted as $\cost(Q)$.  Therefore, our goal is to find $Q\subseteq V$ to maximize the combined function:
\begin{equation}\label{eq:optimization}
g(Q) = {\cov}(Q) - {\cost}(Q),
\end{equation}
where $\cov(Q)$ is the monotone and non-negative submodular function and $\cost(Q)$ is the \emph{sum} of the costs of the
elements in the solution, i.e., a non-negative linear function. 

In the case of influence maximization, the goal is to optimize the expected spread of a product or an idea for a seed of nodes $Q$ minus the cost of picking such nodes.
Similarly, in the team-formation scenario, this would mean that we want to optimize the coverage of the task skills that the experts
in $Q$ cover minus the cost of hiring these experts.  
Finally, in recommender systems the goal could be to maximize the diversity between proposed movies minus the cost of their
distance from a particular year of being produced.

In order to capture the demands of such
application domains we consider two variants of the general problem outlined in Equation~\eqref{eq:optimization}: 
the \emph{constrained} and the \emph{unconstrained}.  
The former refers to cases where the maximum number of elements  we aim to pick is given as part of the input; the
latter finds the optimal number of elements to be picked as part of the solution. 

Although for the constrained version we only discuss the cardinality constraint, our methods extend
to handle general matroid constraints as well. We note
that this version is relevant to scenarios where the elements in the collection are partitioned into groups; 
then, a matroid constraint would impose
an upper bound on the number of elements that can be selected from every group. 

Finally, we also consider the online version of the above problems.  
In this setting the elements in the collection 
become available in an online fashion.  The
constant addition of data in online platforms
makes this setting increasingly relevant.

\spara{Algorithmic contributions:} 
The structure of the combined objective in Equation~\eqref{eq:optimization} in the offline setting and under the cardinality constraint has been addressed by previous works~\cite{feldman2019guess, harshaw2019submodular} as well.
What distinguishes us from these works is that we intentionally design algorithms with slightly weaker theoretical approximation guarantees that allow the use of runtime acceleration techniques, such as lazy greedy evaluations.
As a result, we experimentally show that our proposed algorithms can achieve significant speedups compared to previous approaches in the constrained and unconstrained settings, respectively, while we also show that in practice they perform equally well.
In addition, we propose algorithms with provable guarantees for the online and streaming settings as well as algorithms for more general constraints such as a matroid constraint.

\spara{Experimental results:} We experimentally evaluate our algorithms on real datasets
from a variety of domains such as social networks, crowdsourcing platforms and recommender systems. For our experiments
we  use different instances of the submodular function $\cov$ and the cost function $\cost$ -- chosen appropriately
for the specific application domain we experiment with.
Our experiments show that our algorithms obtain solutions with quality at least as good as existing algorithms while being
significantly faster.
 
\section{Related work}\label{sec:related}

In this section, we highlight the relationship between our work and research done 
in application and theoretical domains.

\spara{Influence maximization:} The seminal work of Kempe {\etal}~\cite{kempe2003maximizing} ignited a lot of subsequent research on influence maximization on social graphs~\cite{li2018influence}. 
In all of these works, there is an underlying information propagation model and a social network that
captures the degree of influence that every node has on others.  
The goal is to identify $k$ nodes in the network to maximize the expected spread of any item (e.g., idea or product) that these nodes adopt.  
Based on this general idea, there has been a set of 
followup works~\cite{borgs2014maximizing, chen2009efficient, leskovec2007cost,li2015real}. All these works 
assume some diffusion model such that the expected spread is a submodular function and therefore a greedy algorithm can be deployed to maximize it. 
Our work generalizes all these works as we want to maximize the expected spread minus the cost for convincing these nodes to adopt the particular item.

\spara{Team formation:} The classic team-formation problem~\cite{anagnostopoulos18algorithms,anagnostopoulos2012online,bhowmik2014submodularity,kargar2013finding,lappas2009finding,yin2018social} 
assumes that there is a pool of experts and a subset of them is selected to cover the requirements of a task, while some  criteria related to the team functionality 
(e.g., communication cost as captured in their collaboration network) are optimized. 
At the heart of all team-formation problems defined today is a \emph{set cover} problem where the goal is to cover the skills of the input task. 
In our formulation we consider an extension of this setting where not all skills need to be covered
and we seek to strike a balance between the covered skills and the cost of building a team.

\spara{Recommender systems:} Recommendations in recommender systems have often been formulated as a submodular function optimization problem, with the goal being to maximize coverage (e.g., of product attributes being addressed by reviews), while maximizing the diversity of the items being recommended~\cite{ashkan2014diversified, borodin2012max, parambath2018saga}, or improving the summarization of a set of items (e.g., restaurants available)~\cite{dasgupta2013summarization, mirzasoleiman2016fast}.
In all of the above cases, the goal has been to maximize a submodular function and not the combined function of the benefit
minus the cost associated with these recommendations. Kazemi {\etal}~\cite{kazemi2020regularized} model recommender systems as maximizing benefit minus cost, and we consider their applications in our experimental evaluation. We discuss these applications in Section~\ref{sec:preliminaries}.

\spara{Maximizing submodularity minus cost:} 
Several algorithms have been developed for maximizing both monotone and general submodular functions, but they achieve provable approximation guarantees only for non-negative functions, whereas the objective in Equation~\eqref{eq:optimization} is potentially negative. Existing hardness results imply that no multiplicative approximation guarantees are possible in polynomial time for maximizing a potentially negative submodular function with or without constraints \cite{papadimitriou1991optimization,Feige1998}.\footnote{One can observe that it is \textbf{NP}-hard to decide whether the optimum value of a submodular objective is positive or not, since we could use such a subroutine and binary search over the optimum value to obtain arbitrarily good approximate solutions, which contradicts existing hardness of approximation results for problems such as maximum cut and maximum coverage \cite{papadimitriou1991optimization,Feige1998}.} Nevertheless, the objective function we consider has some structure that has been exploited  in previous works~\cite{feldman2019guess,harshaw2019submodular,sviridenko2017optimal}.  These works have shown that in this case we should aim for a weaker notion of  approximation and find a solution $Q$ satisfying
\[
\cov(Q)-\cost(Q)\geq \alpha \cdot \cov(\opt)-\cost(\opt),
\]
for some $\alpha\leq 1$. 
The aforementioned works  propose algorithms that achieve $\alpha=(1-1/e)$ in the offline setting, which is the best guarantee achievable in polynomial time for a cardinality constraint~\cite{Nemhauser1978,Feige1998}. One of the main downsides of these algorithms is that the running time can be prohibitive. The works \cite{sviridenko2017optimal,feldman2019guess} propose algorithms based on the continuous greedy algorithm that maximizes the multilinear extension, a continuous function extending the submodular function to the domain $[0, 1]^n$. The multilinear extension is expensive to evaluate and the continuous greedy algorithm requires many iterations to converge. As a result, the algorithms have very high running times. The algorithm of \cite{feldman2019guess} is based on the standard discrete greedy algorithm, which is much more efficient, but it applies the greedy approach to a distorted objective that changes throughout the algorithm and we cannot use techniques such as lazy evaluations to speed up the algorithm. Thus the running time of the algorithm of  \cite{feldman2019guess} is $\Theta(n^2)$, where $n$ is the size of the ground set, whereas the standard greedy algorithm can be implemented to run in nearly-linear time using approximate lazy evaluations. Moreover, the implementation of the standard greedy algorithm using exact lazy evaluations achieves significant speedups in practice without affecting the approximation guarantee \cite{minoux1978accelerated}. Additionally, these algorithms are only for the offline problem.

In this paper, we give a novel approach that overcomes these limitations. Our main insight is very simple but very effective: instead of maximizing the original objective $ \cov(Q)-\cost(Q)$, we maximize a \emph{scaled} objective $ \cov(Q)- s \cdot \cost(Q)$, where $s > 1$ is an absolute constant. Unlike the approach of \cite{feldman2019guess}, our objective does not change throughout the algorithm and thus we can use lazy evaluations. Moreover, we can leverage a wide-range of existing algorithmic approaches, such as the standard greedy algorithm in the offline setting and variants of greedy in the online setting. As a result, we obtain faster offline algorithms for the cardinality-constrained problem, algorithms for the online and streaming settings, and algorithms for more general constraints such as a matroid constraint.

Our problem formulation can be viewed as a Lagrangian relaxation of the problem of maximizing a submodular function subject to a  knapsack constraint. Several algorithms have been proposed for the latter problem, including algorithms achieving the optimal $1-1/e$ approximation guarantee~\cite{sviridenko2004note}. However, these algorithms have very high running times and they are primarily of theoretical interest. For example, the algorithm of \cite{sviridenko2004note} has running time $\Theta(n^5)$, where $n$ is the size of the ground set. Thus, even if we used lazy evaluations to speed up the algorithm, the enumeration will still be a significant bottleneck. Obtaining fast and practical algorithms for the knapsack problem remains an outstanding open problem (see e.g. \cite{ene2019nearly} and references therein). Thus, our problem formulation also comes with significant algorithmic benefits compared to the formulation with a hard budget constraint.

In contemporaneous work, Kazemi  {\etal}~\cite{kazemi2020regularized} develop streaming and distributed algorithms for the cardinality-constrained problem. The streaming algorithms developed in their work and ours are conceptually very similar and they achieve the same approximation guarantee.

\section{Problem definition}\label{sec:preliminaries}
Throughout the paper we will assume a set of $n$ elements $V = \{1,\ldots, n\}$. We also assume two functions:
$\cov:2^V\rightarrow \mathbb{R}$ and $\cost:2^V\rightarrow \mathbb{R}$, such that 
$\cov$ is a monotone and non-negative submodular function  and $\cost$ is a non-negative linear function. 

Recall that a set function $h: 2^V \rightarrow \mathbb{R}$ is \emph{monotone} if
\begin{align*}
h(S) \leq h(T) \quad \forall S \subseteq T\subseteq V
\end{align*}
The set function $h: 2^V \rightarrow \mathbb{R}$ is \emph{submodular} if it satisfies the following diminishing returns property:
\begin{align*}
h(T\cup \{u\}) - h(T) \leq h(S\cup \{u\}) - h(S) \quad \forall S \subseteq T, u \in V \setminus T 
\end{align*}
An equivalent definition of submodularity is the following:
\begin{align*}
h(S) + h(T) \geq h(S\cap T) + h(S \cup T) \quad \forall S, T \subseteq V
\end{align*}
Given the above, we define our \emph{objective function} $g:2^V\rightarrow \mathbb{R}$ as follows:
\begin{equation}\label{eq:objective}
g(Q) = \lambda\cdot \cov(Q) - \cost(Q).
\end{equation}
Note that function $g$ is also submodular but it can take both negative and positive values and it is not monotone.  
The problems we aim to solve in this paper are related to optimizing this function and can be defined as follows.

\begin{problem}[{\constrained}]\label{problem:constrained}
Given a set of elements $V$ and an integer $k$, find $Q\subseteq V$ such that
$|Q|\leq k$ and
\begin{equation}
g(Q) = \lambda \cdot \cov(Q) - \cost(Q)
\label{pb:kproblem}
\end{equation}
is maximized.
\end{problem}

Our approach extends to a general matroid constraint. The algorithm we discuss in Section~\ref{sec:alg-cardinality} can be minimally modified to work 
for general matroid constraints, including its running time and approximation bounds. 

We also consider the unconstrained version of the above problem, which is defined as follows. 

\begin{problem}[{\unconstrained}]\label{problem:unconstrained}
Given a set of elements $V$ find $Q\subseteq V$ such that
\begin{equation}
g(Q) = \lambda\cdot \cov(Q) - \cost(Q)
\end{equation}
is maximized.
\end{problem}

\spara{Online problems:}
In addition to the offline setting, we study the above problems in online and streaming models of computation. We consider Problem~\ref{problem:unconstrained} in the online model where the ground set elements arrive in an online fashion, one at a time, in an arbitrary (adversarial) order. When an element arrives, we need to decide whether to add it to the solution, and this decision is irrevocable. We refer to this problem as {\unconstrainedonline}.

We also consider Problem~\ref{problem:constrained} in the streaming model where the elements arrive one at a time as in the online setting but we are allowed to store a small set of elements in memory and select the final solution from this set. We refer to this problem as {\constrainedonline}.  

\spara{The normalization coefficient $\lambda$:}
In the above definitions, $\lambda$ is a normalization coefficient that encodes our bias between the prizes and the costs. 
One can also think of $\lambda$ as a way to convert the two quantities into the same units.
Determining its value is application-dependent and is discussed in Section~\ref{sec:picklambda}.
Our algorithmic analysis is independent of this coefficient, and therefore from now on we will use
$\cov(Q)$ to refer to $\lambda\cdot\cov(Q)$.
We will also refer to 
$g(Q) = \cov(Q) - \cost(Q)$ as the \emph{combined objective function}.

\spara{Approximation guarantees:} Note that while function $\cov$ is monotone submodular and non-negative,  the combined objective function $g$ is a \emph{potentially negative} submodular function. As discussed in the introduction, no multiplicative factor approximation is possible for the problem of maximizing a submodular function that is potentially negative. Similarly to previous work (see Section~\ref{sec:related}), our algorithms construct solutions with the following kind of weaker approximation guarantees: 
\begin{equation*}
\cov(Q)-\cost(Q)\geq \alpha\; \cov(\opt)-\cost(\opt),
\end{equation*}
where $\opt$ is an optimal solution to the problem and $\alpha\leq 1$.   

\spara{Problem instances:}  In our experimental evaluation, we consider several instantiations of the monotone submodular function $\cov$ and the linear function $\cost$, arising in influence maximization in social networks, team formation, and recommender systems. 

\emph{Influence maximization in social networks:} The ground set $V$ corresponds to social network nodes and the
goal is to pick a subset of the nodes $Q\subseteq V$ such that the spread of a product (or an idea) in the network is maximized.
The function $\cov(Q)$ corresponds to the expected number of people that adopt the product given seed set $Q$.
The way the expectation is computed depends on the \emph{information propagation model} being used.
In this paper, we focus on the independent cascade and linear-threshold models~\cite{kempe2003maximizing} which makes $\cov(Q)$ submodular.
We also use a non-negative linear function $\cost(Q)$ to 
quantify the sum of the costs of convincing each node to adopt a product. 
Assuming that influential nodes are more expensive to convince, in our experiments, we model  the cost for convincing each individual as being proportional to the node's degree.

\emph{Recommender systems:} Kazemi {\etal}~\cite{kazemi2020regularized} consider several applications to recommender systems and show that they can be modeled as instances of the problem {\constrained}.  In our experimental evaluation, we use the following two problem instances proposed by them, which we describe here for completeness. In both applications, the ground set $V$ corresponds to items --- e.g., restaurants or movies --- and each item $i$ is associated with a set of features which are then used to compute the distance between two items $d(i,j)$. A similarity matrix $\mathbf{M}$ between items is formed by setting $\mathbf{M}(i,j) = e^{-d(i,j)}$.

The first application considers restaurant recommendations. The items are restaurants. The submodular function {\cov} is defined as:
\begin{equation}\label{eq:location}
\cov(Q) = \sum_{i=1}^n\max_{j\in Q} \mathbf{M}(i,j).
\end{equation}
For the linear cost function $\cost(Q) = \sum_{i\in Q} c_i$, the cost $c_i$ corresponds to the distance of restaurant $i$ to the center of the city.

The second application considers movie recommendations. The items are movies. The submodular function {\cov} is defined as:
\begin{equation}\label{eq:determinant}
\cov (Q)= \log \det (\mathbf{I} + \alpha \mathbf{M}_Q),
\end{equation}
where $\mathbf{M}_Q$ is the principal submatrix of $\mathbf{M}$ indexed by $Q$, $\mathbf{I}$ is the identity matrix and $\alpha$ is a positive scalar. 
Informally, this objective aims to diversify the vectors in $Q$.
For the linear cost function $\cost(Q) = \sum_{i\in Q} c_i$, the costs are given by $c_i$ = 10 - rating$_i$, where rating$_i$ denotes the average rating that movie $i$ has received.

\emph{Team formation:} The ground set $V$ is a set of experts and each expert $i$ is associated with 
a set of skills $S_i\subseteq S$, where $S$ is  a universe of skills. 
Given a task $T\subseteq S$ and a set of experts $Q\subseteq V$ we define the \emph{coverage} function to be the number of skills in $T$ that is covered by at least one expert in $Q$. Thus
 \begin{equation}\label{eq:skills}
 \cov(Q) = \left | \left(\cup_{i\in Q}S_i\right)\cap T \right|.
 \end{equation}
Each expert $i$ is also associated with a \emph{cost} $c_i$ needed to hire the expert. The cost of hiring a team is the sum of the expert costs, i.e., $\cost(Q) = \sum_{i\in Q} c_i$.

\section{The {\tt Cost-Scaled Greedy} Algorithm}\label{sec:alg-cardinality}

\begin{algorithm}[t]
\begin{flushleft}
\textbf{Input:} Ground set $V$, scaled objective $\tilde{g}(Q)=\cov(Q)-2\cost(Q)$, cardinality $k$. \\
\textbf{Output:} Solution $Q$.
\end{flushleft}
\begin{algorithmic}[1]
\STATE $Q\gets\emptyset$
  \FOR{$i=1,\ldots,k$}
    \STATE $e_{i} = \arg\max_{e\in V}\tilde{g}(e|Q)$\label{line:condition}
    \IF{$\tilde{g}(e_i \vert Q) \leq 0$}
	    \STATE break \label{line:terminate}
		\ENDIF
	  \STATE $Q\gets Q\cup\{e_{i}\}$
  \ENDFOR
\RETURN{$Q$}
\end{algorithmic}
\caption{\label{algo:cardinality}The {\costscaledgreedy} algorithm for the cardinality-constrained problem {\constrained}.}
\end{algorithm}

In this section, we consider the cardinality-constrained problem {\constrained}. Our algorithm for this problem is shown in Algorithm~\ref{algo:cardinality}. Throughout the paper, by elements we mean the elements of the ground set $V$. For a set function $h$, we use the notation $h(e\vert Q):=h(Q\cup\{e\})-h(Q)$ to denote the marginal gain of $e$ on top of $Q$. Our approach is very simple but effective: we apply the standard Greedy algorithm to the scaled objective $\tilde{g}(Q) = \cov(Q)- 2\cost(Q)$, and we stop adding elements once the marginal gains become negative (line~\ref{line:terminate}).  As we discuss below, the algorithm can be implemented using lazy evaluations, which leads to a very efficient and practical algorithm.  In Appendix~\ref{app:offline}, we show the following guarantee:
\begin{thm}
\label{thm:cardinality}
Algorithm \ref{algo:cardinality} returns a solution $Q$ of size at most $k$ satisfying $\cov(Q)-\cost(Q)\geq\frac{1}{2}\cov(\opt)-\cost(\opt)$.
\end{thm}

\spara{Extension to matroid constraints:} Similarly to the standard Greedy algorithm, our algorithm and its analysis readily extends to the more general matroid-constrained problem. We give this extension in Section~\ref{app:matroid}. The matroid algorithm can also be sped up using the lazy evaluation technique as described below. 

\spara{Running time  and speedups:}
Similarly to the standard Greedy algorithm, the running time of {\costscaledgreedy} is $O(nk)$ evaluations of the functions $\cov$ and $\cost$, where $n$ is the number of elements and $k$ is the cardinality constraint: there are $k$ iterations and, in each iteration, we spend $O(n)$ function evaluations to compute all of the marginal gains and find the element with maximum marginal gain.  

The computational bottleneck of the algorithm is in finding the element with maximum marginal gain $\tilde{g}(e|Q)$ in every iteration. To speed up these computations and avoid unnecessary evaluations, we 
deploy the lazy evaluations technique introduced by Minoux \cite{minoux1978accelerated} for the standard Greedy algorithm.
That is, we 
store each element in a maximum priority queue with a key $v(e)$. We initialize the keys to $v(e)=\tilde{g}(e|\emptyset)$. The keys are storing potentially outdaded marginal gains and the algorithm updates them in a lazy fashion. Since $\tilde{g}$ is submodular, marginal gains can only decrease as the solution $Q$ grows and, as a result, the keys are always an upper bound on the corresponding marginal gains. In each iteration, the algorithm uses the queue to find the element with maximum marginal gain as follows. We remove from the queue the element $e$ with maximum key and evaluate its marginal gain $\tilde{g}(e|Q)$ with respect to the current solution $Q$. We then compare the marginal gain  $\tilde{g}(e|Q)$ to the key $v(e')$ of the element $e'$ that is now at the top of the queue (before removing $e$ from the queue, $e'$ was the element with the second-largest key). If $\tilde{g}(e|Q)\geq v(e')$, then $e$ is the element with largest marginal gain, since the key of every element is an upper bound on its current marginal gain. Otherwise, we reinsert $e$ into the queue with key $\tilde{g}(e|Q)$ and repeat. 

We use {\costscaledlazygreedy} to refer to the implementation of {\costscaledgreedy} with lazy evaluations. The correctness of {\costscaledlazygreedy} follows directly from submodularity, and the solution constructed is the same as that of {\costscaledgreedy}. While the worst-case running time of {\costscaledlazygreedy} and {\costscaledgreedy} are the same, lazy evaluations lead to significant speedups in practice.

We note that there is also an approximate version of the lazy evaluations technique that allows us to obtain worst-case running time that is nearly-linear at a small loss in the approximation guarantee \cite{badanidiyuru2014fast}. We do not consider this variant in this paper.

\section{The {\tt online} algorithm}
 \label{sec:alg-online}

We now turn our attention to the {\unconstrained} problem  in the \emph{online} model where the elements arrive one at a time. When an element arrives, we need to decide whether to add it to the solution, and this decision is irrevocable.  Algorithm~\ref{algo:online} considers the scaled objective $\tilde{g}(Q)=\cov(Q)-2\cost(Q)$ and it accepts every element that has positive marginal gain with respect to this scaled objective. The following theorem states our approximation guarantee. We give the analysis in Section~\ref{app:online}.

\begin{algorithm}[t]
	\begin{flushleft}
	\textbf{Input:} Stream of elements $V$, scaled objective $\tilde{g} = \cov - 2\cost$. \\
	\textbf{Output:} Solution $Q$.
	\end{flushleft}
	\begin{algorithmic}[1]
	\STATE $Q\gets\emptyset$
        \FOR{each arriving element $e$}
        		\IF{$\tilde{g}(e|Q)>0$}
			\STATE $Q\gets Q\cup\{e\}$
		\ENDIF
        \ENDFOR
        \RETURN{$Q$}
      \end{algorithmic}
	\caption{\label{algo:online}The {\onlinecostscaledgreedy} algorithm.}
\end{algorithm}

\begin{thm}
Algorithm~\ref{algo:online} returns a solution $Q$ satisfying $\cov(Q)-\cost(Q)\geq\frac{1}{2}\cov(\opt)-\cost(\opt)$.
\end{thm}

\section{The {\tt streaming} algorithm} 
\label{sec:alg-streaming}

This section considers the  {\constrained} problem in the streaming model. The algorithm is an extension of the online algorithm from Section~\ref{sec:alg-online}. As before, we consider the scaled objective $\tilde{g}(Q)=\cov(Q)-\const\cdot \cost(Q)$, where $\const\geq1$ is an absolute constant (the right choice for $\const$ is no longer $2$, see Theorem~\ref{thm:streaming} below). Now, instead of picking elements whose scaled marginal gain is positive, we pick elements whose scaled marginal gain is above a suitable threshold. In other words, we apply the single-threshold Greedy algorithm \cite{badanidiyuru2014streaming,kumar2015fast} to the scaled objective.  The resulting algorithm is shown in Algorithm~\ref{algo:streaming}. The following theorem shows that there is a way to set $\tau$ and $\const$ so that Algorithm~\ref{algo:streaming} returns a good approximate solution, and we give its proof in Appendix~\ref{app:streaming}.

\begin{algorithm}[t]
	\begin{flushleft}
	\textbf{Input:} Stream of elements $V$, scaled objective $\tilde{g} = \cov - \const \cdot \cost$ ($\const \geq 1$ is an absolute constant), cardinality $k$, threshold $\tau$. \\
	\textbf{Output:} Solution $Q$.
	\end{flushleft}
	\begin{algorithmic}[1]
	\STATE $Q\gets\emptyset$
        \WHILE{stream not empty}
        		\STATE $e\gets$next stream element
        		\IF{$\tilde{g}(e|Q)\geq\tau$ and $|Q|<k$}
			\STATE $Q\gets Q\cup\{e\}$
		\ENDIF
        \ENDWHILE
        \RETURN{$Q$}
      \end{algorithmic}
	\caption{\label{algo:streaming}The {\onlinekcostscaledgreedy} algorithm.}
\end{algorithm}

\begin{thm}
\label{thm:streaming}
When run with scaling constant $\const=\frac{1}{2}\left(3+\sqrt{5}\right)$ and threshold $\tau=\frac{1}{k}\left(\frac{1}{2}(3-\sqrt{5})\cov(\opt)-\cost(\opt)\right)$, Algorithm \ref{algo:streaming} returns a solution $Q$ such that $|Q| \leq k$ and
\[ \cov(Q)-\cost(Q)\geq \frac{1}{2}\left(3-\sqrt{5}\right)\cov(\opt)-\cost(\opt) \]
\end{thm}

Setting the threshold as suggested by the above theorem requires knowing $\hat{g}(\opt)$, where $\hat{g}(Q):=\frac{1}{2}(3-\sqrt{5})\cov(Q)-\cost(Q)$. To remove this assumption, we use the standard technique introduced by \cite{badanidiyuru2014streaming} and run several copies of the basic algorithm in parallel with different guesses for $\hat{g}(\opt)$. We only lose $\epsilon$ in the approximation due to guessing and we use $O\left(k \log{k} / \epsilon \right)$ total space to store the $O(\log{k}/\epsilon)$ solutions.  We refer the reader to Appendix~\ref{app:streaming} for more details.

\begin{thm}
There is a streaming algorithm {\onlinekcostscaledgreedy} for the cardinality-constrained problem $\max_{|Q| \leq k} \cov(Q) - \cost(Q)$ that takes as input any $\epsilon > 0$ and it returns a solution $Q$ satisfying
\[ \cov(Q)-\cost(Q)\geq \left(\frac{1}{2}\left(3-\sqrt{5}\right) - \epsilon \right)\cov(\opt)-\cost(\opt) \]
The algorithm uses $O\left(k \log{k} / \epsilon \right)$ space.
\end{thm}

The result presented in this section was obtained independently and concurrently by \cite{kazemi2020regularized} and a preliminary version of our work \cite{ene2020team}.

\section{Experiments}
\label{sec:experiments}

In this section, we experimentally evaluate our algorithms on real-world datasets. 
\subsection{Algorithms of the comparative study}
We compare our algorithms to a variety of baseline methods, as well as to the state-of-the-art algorithms.

\spara{Algorithms used for the {\constrained} problem:}
We experimentally evaluate our main algorithms for the {\constrained} problem: the cost-scaled Greedy algorithm ({\costscaledgreedy}) and its variant with lazy evaluations ({\costscaledlazygreedy}) (see Section~\ref{sec:alg-cardinality}). Recall that the two algorithms return the same solution, but {\costscaledlazygreedy} is expected to be significantly faster due to the lazy evaluations. Thus, in the objective evaluation plots we only include  {\costscaledlazygreedy}. However, in the running-time evaluation plots we present both algorithms to demonstrate the speedups achieved due to lazy evaluations.

We also  evaluate our streaming algorithm, {\onlinekcostscaledgreedy}, described in Section~\ref{sec:alg-streaming}. The algorithm is designed for the more challenging setting where the elements arrive in a stream, one at a time, and the algorithm can make only one pass over the elements and store only a small number of elements in memory. We evaluate the  performance
of  {\onlinekcostscaledgreedy} against offline algorithms that  have complete knowledge of the input datasets and can make many passes over the elements.

The baselines we consider are algorithms proposed in prior work 
as well as some intuitive heuristics. We list them below:
\squishlist
\item
\emph{{\distortedgreedy}~\cite{harshaw2019submodular}}: Builds on the Greedy approach but, instead of considering a constant scaled objective like we do, the authors design a \emph{distorted objective} which changes throughout the algorithm. The distorted objective initially places higher relative importance on the modular cost term $c$, and gradually increases the relative importance of the coverage function as the algorithm progresses. {\distortedgreedy} makes O($nk$) evaluations and returns a solution $Q$ of size at most $k$ satisfying $\cov(Q)-\cost(Q)\geq(1-\frac{1}{e})\cov(\opt)-\cost(\opt)$.

\item
\emph{{\stochasticdistortedgreedy}~\cite{harshaw2019submodular}}: Uses the same distorted objective as {\distortedgreedy} but has faster asymptotic runtime because it optimizes over a random sample in each iteration.

\item
\emph{{\greedy}}: This is the greedy algorithm for maximizing submodular set functions~\cite{nemhauser1978analysis}. The only difference from {\costscaledgreedy} is that instead of computing the marginal value with respect to the scaled objective we use the original objective $g = \cov - \cost$.

\item
\emph{{\topkexperts}}: This is a natural heuristic baseline algorithm that runs as follows. The algorithm gives each element $e\in V$ a linear weight $w(e)= \cov(\{e\})-\cost(e)$ and it selects the (at most) $k$ elements with the largest positive weights. If there are fewer than $k$ elements with positive weight, the algorithm selects all of the elements with positive weights; otherwise, the algorithm selects the $k$ elements with largest weights. 
\squishend

\spara{Algorithms used for the {\unconstrained} problem:} For the {\unconstrained} problem
we evaluate all the methods that we used for the {\constrained} problem by setting $k=n$. So all algorithms
described above are included in the comparison.

Additionally, we evaluate our online algorithm, {\onlinecostscaledgreedy}, described in Section~\ref{sec:alg-online}.
This algorithm addresses the harder online problem where the elements are presented in an online fashion, and the algorithm needs to irrevocably decide whether to include the element in the solution when the element arrives. We evaluate the algorithm's performance against offline algorithms that have complete knowledge of the input datasets.

In terms of baselines, we also consider the following:
\squishlist
\item 
{\unconstraineddistortedgreedy}~\cite{harshaw2019submodular}: A linear-time algorithm for the unconstrained problem that runs for $n$ iterations and in each iteration evaluates the marginal gain of a single element sampled uniformly at random.
\squishend

\subsection{Experimental setup}
For all our experiments we evaluate the algorithmic performances on different subsets of the original data and we report the average performance value of each algorithm over these 15 samples, denoted as its line, as well as the confidence interval of the result, denoted as the bar around the line.
Our code is in Python and for all our experiments we use single-process implementations on a 64-bit MacBook Pro with an Intel Core i7 CPU at 2.6GHz and 16 GB RAM.
For replication purposes we make the code, the datasets and the chosen hyperparameters available online.\footnote{https://www.dropbox.com/sh/vu87zte0p4hrybz/AACs0liWgCejxj5R9FEowhRza?dl=0}

\spara{Selecting the value of the normalization coefficient $\lambda$:}
\label{sec:picklambda}
The combined objective introduced in Section \ref{sec:preliminaries} compares the gain and the cost;
the gain corresponds to the value of a submodular function and the cost is the value of a linear function.
The purpose of the parameter $\lambda$ is to transform these two quantities into comparable units and we set it as follows.
First,  we use the well-known greedy algorithm \cite{nemhauser1978analysis} to find the set of elements $Q^\ast$ that maximize the submodular function.
Then, we define $\lambda=\beta\frac{\cost(Q^\ast)}{\cov(Q^\ast)}$, with $\beta \in \{2, 4\}$ depending on the dataset such that the gain and the cost are in a comparable scale.

\spara{Setting the algorithmic parameter $\epsilon$:}
Recall that algorithms {\stochasticdistortedgreedy} and {\onlinekcostscaledgreedy} require an error parameter $\epsilon$ as part of their input.
This is a trade-off parameter between the quality of the solution and the algorithm's running time.
To select an appropriate value  $\epsilon$ we performed a set of experiments for different $\epsilon$ values and picked the one that achieves the best (for the algorithm) solution, without sacrificing the running time.
Due to lack of space we omit these plots. 
Throughout the experiments we fix $\epsilon$=0.01 and $\epsilon$=0.05 for  {\stochasticdistortedgreedy} and {\onlinekcostscaledgreedy}  respectively.

\spara{Applications and datasets:} We experimentally evaluate
 the proposed algorithms on datasets from 
application domains we discussed in Section~\ref{sec:preliminaries}.  

\emph{Influence Maximization:} In these experiments, we follow the experimental setup of \cite{chen2009efficient, kempe2003maximizing}.  We use the academic collaboration network from  the ``High Energy Physics-Theory'' section of the e-print arXiv that was used in these prior works.
We consider a network which contains 1077 nodes and 3505 edges (one of the largest components of the whole dataset) and we refer to this as the
{\InfluenceDataset} dataset.
In our experiments,
we use  the independent cascade model,  but similar results hold for the linear threshold model as well.
We treat the multiplicity of edges as weights and similar to the experiments of~\cite{chen2009efficient,kempe2003maximizing} we assign a uniform probability of $p=0.01$ to each edge.
The specific instantiations of functions $\cov$ and $\cost$ in this application are the ones described in Section~\ref{sec:preliminaries}.

\emph{Team Formation:} In these experiments, we follow the experimental setup of \cite{anagnostopoulos18algorithms, golshan14profit}.
We use a real-world dataset from the online expertise-management platform {\texttt{guru.com}} that henceforth we refer to as {\GuruDataset}. 
All expert-related data used in this work are obtained from anonymized profiles of members registered in the marketplace. 
Our dataset has 6120 experts and 20 tasks, each requiring 15 skills. For each expert we also know the set of skills
the expert has.
Since we consider multiple tasks, we evaluate our algorithms for each task separately and report the average performance value of each algorithm over all tasks as well as the confidence interval of the results.
The  instantiations of functions $\cov$ and $\cost$ in this application are the ones described in Section~\ref{sec:preliminaries} (see Eq.~\eqref{eq:skills}).

\emph{Recommender Systems:} 
For the recommender systems application we consider two separate applications from the work of \cite{kazemi2020regularized}:
(i) restaurant summarization, and (ii) movie recommendation. We follow the experimental setup of \cite{kazemi2020regularized}.

For  the restaurant summarization task we use a subset of restaurant businesses obtained from the Yelp Academic dataset.\footnote{\url{https://www.yelp.com/dataset}}
We use features that cover a range of restaurant attributes.\footnote{To extract the features of each restaurant we use the script provided at \url{https://github.com/vc1492a/Yelp-Challenge-Dataset}.}
We focus on the restaurants from the metropolitan area of Las Vegas and the goal is 
to identify representative restaurants in that area. 
To evaluate our algorithms we consider random restaurant subsets of size 651 (5\% of the data) and report the average performance of our algorithms over all subsets as well as their confidence interval. We refer to this dataset as {\YelpDataset}. The  instantiations of functions $\cov$ and $\cost$ in this application are the ones described in Section~\ref{sec:preliminaries} (see Eq.~\eqref{eq:location}).

For the movie recommendation task we use the MovieLens~\cite{harper2015movielens} dataset.
In this dataset we are given user ratings for different movies.
To extract the feature vector for each movie we use gradient descent~\cite{koren2009matrix}.
At the end of this process we obtain 40 latent factors for each movie.
After filtering out movies with less than $50$ ratings , we consider random subsets of 658 movies to evaluate our algorithms over different inputs. We refer to this dataset as {\MovielensDataset}.
The instantiations of functions $\cov$ and $\cost$ in this application are the ones described in Section~\ref{sec:preliminaries} (see Eq.~\eqref{eq:determinant}).

\subsection{Evaluation for \large{{\constrained}}}
\label{exp:constrained}
\begin{figure*}
\centering
\begin{minipage}[t]{0.24\linewidth}
    \includegraphics[width=\linewidth]{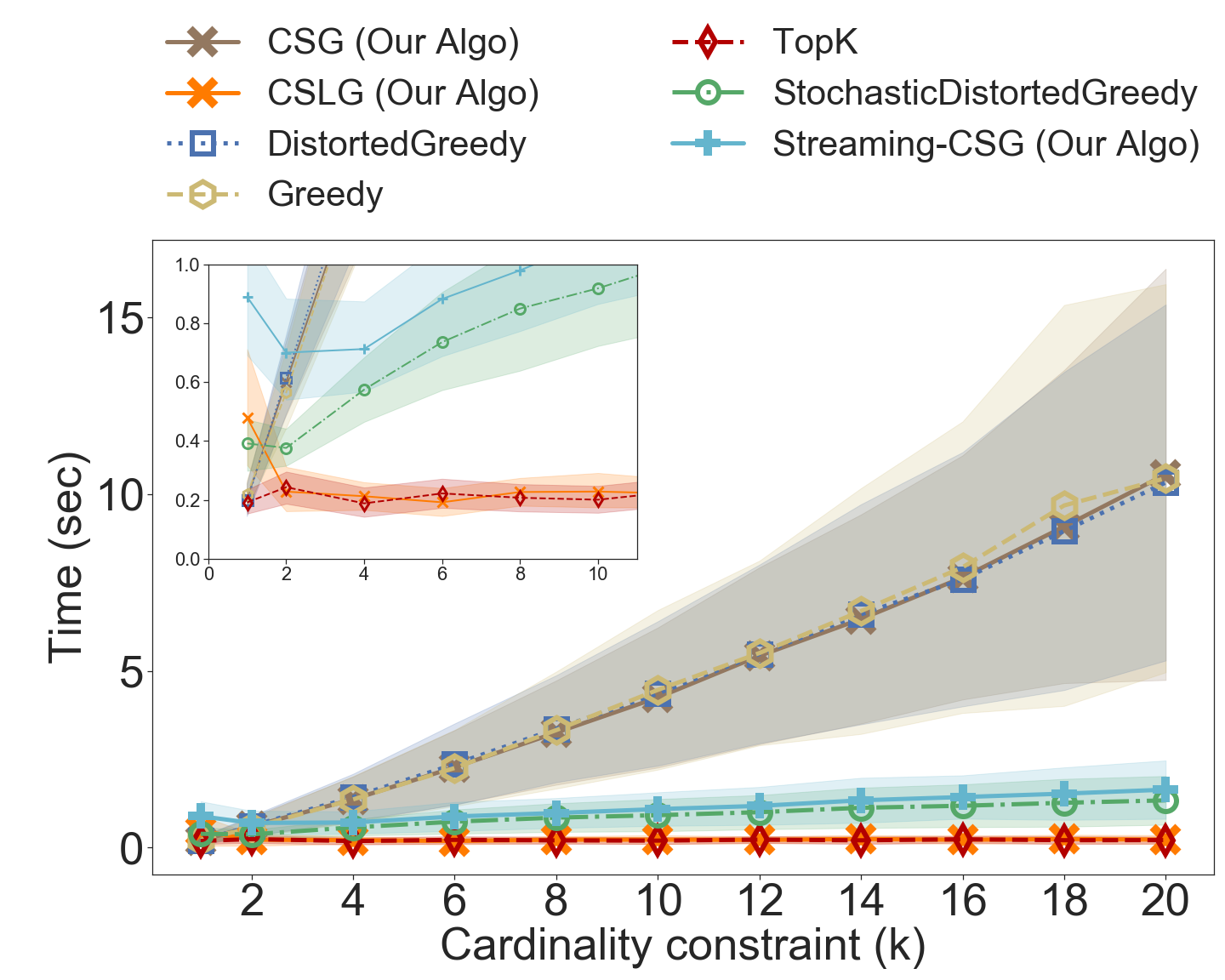}
    \subcaption{\label{fig:time_inf}\scriptsize{{\InfluenceDataset}}}
\end{minipage}
\begin{minipage}[t]{0.24\linewidth}
    \includegraphics[width=\linewidth]{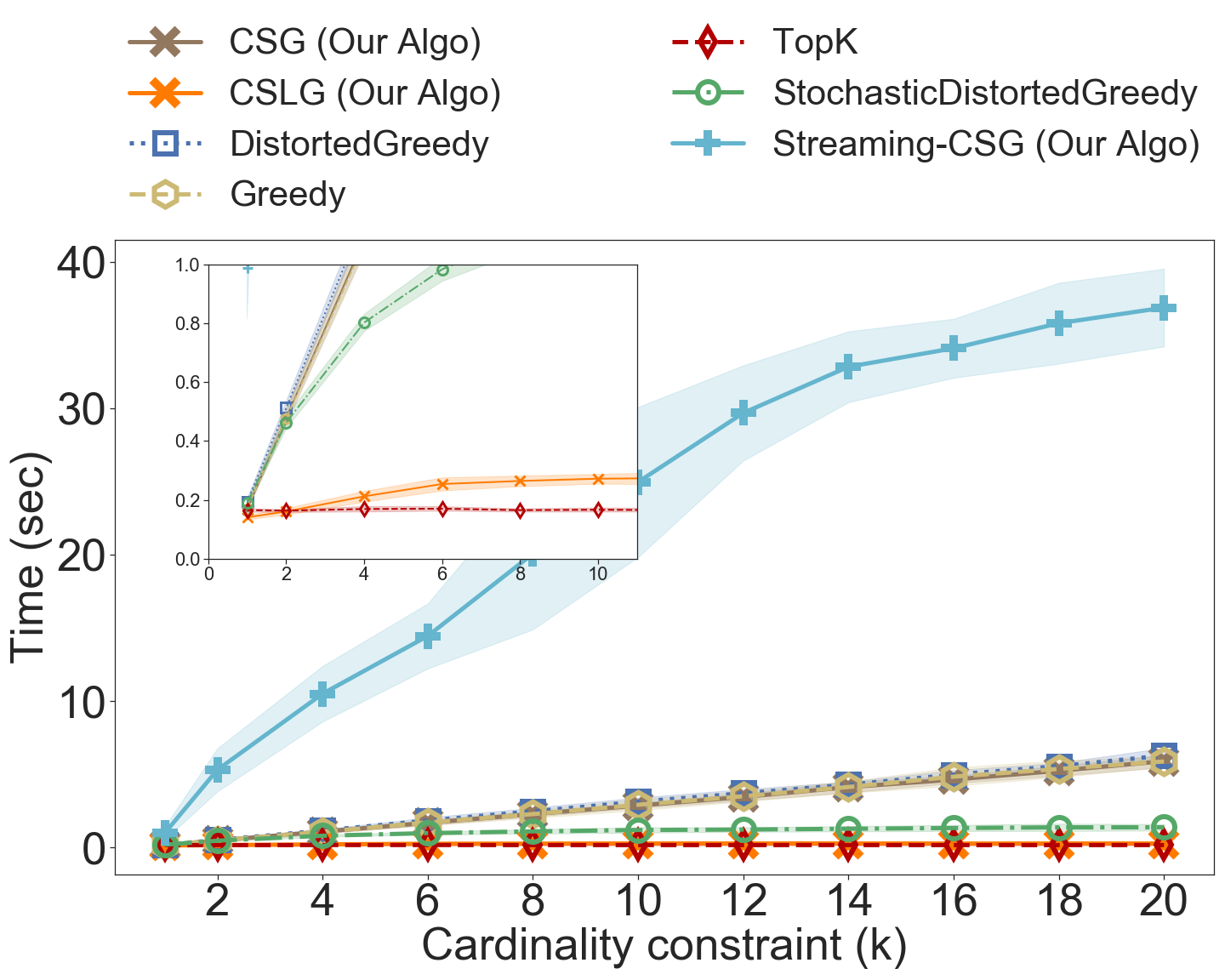}
    \subcaption{\label{fig:time_guru}\scriptsize{{\GuruDataset}}}
\end{minipage}
\begin{minipage}[t]{0.24\linewidth}
    \includegraphics[width=\linewidth]{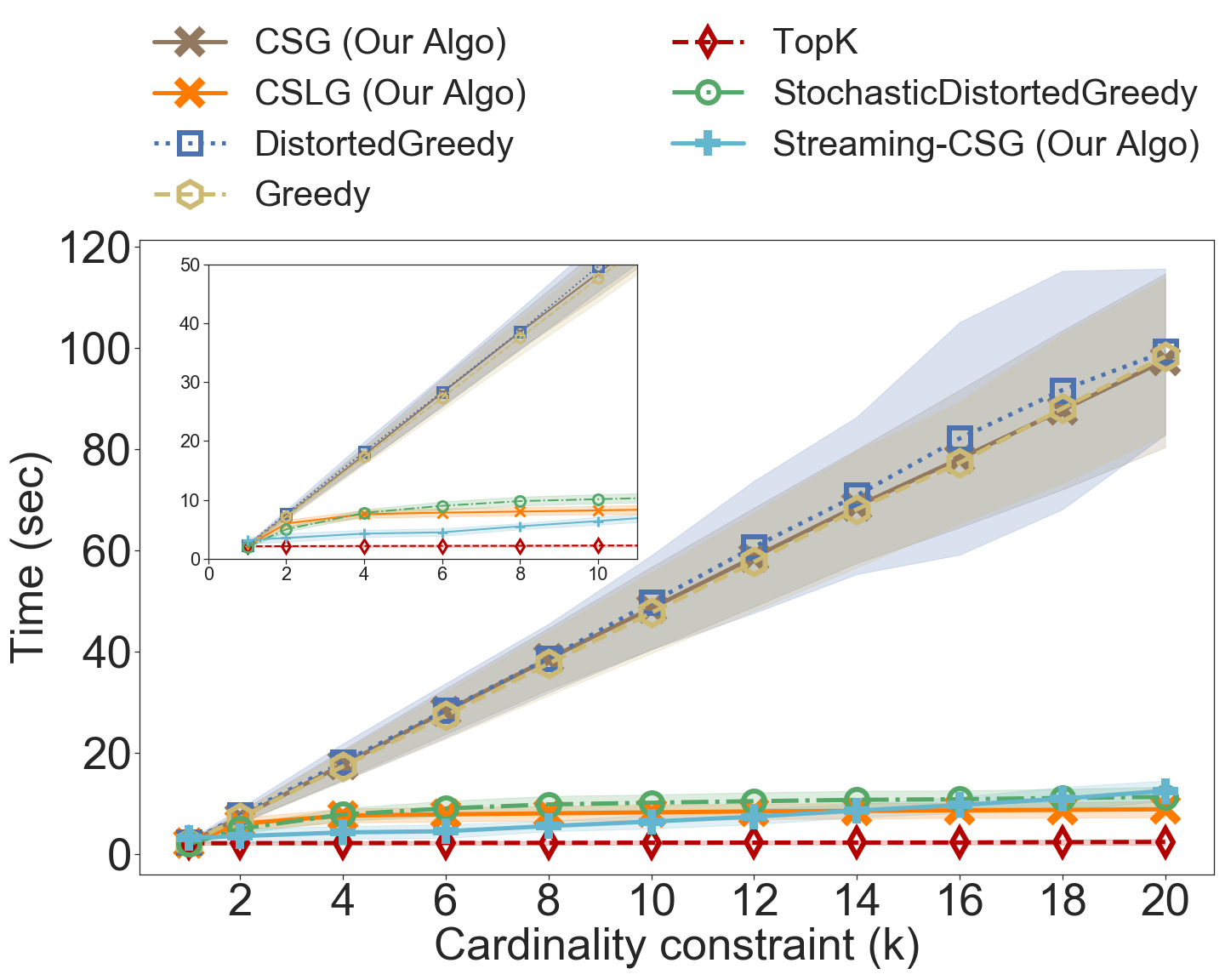}
    \subcaption{\label{fig:time_yelp}\scriptsize{{\YelpDataset}}}
\end{minipage}
\begin{minipage}[t]{0.24\linewidth}
    \includegraphics[width=\linewidth]{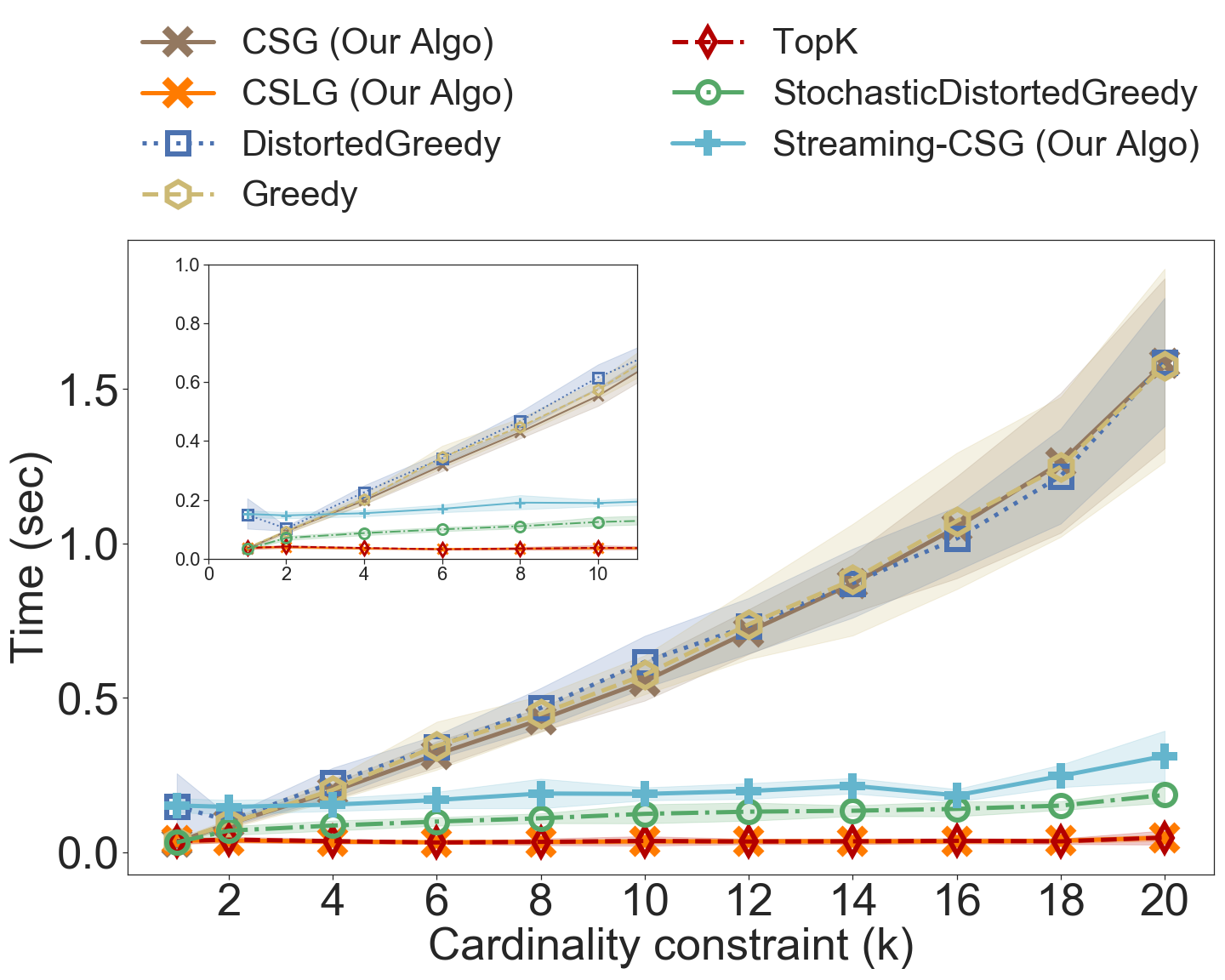}
    \subcaption{\label{fig:time_movie}\scriptsize{{\MovielensDataset}}}
\end{minipage}
\caption{\label{fig:time_constrained} Running time (sec) comparisons of all algorithms for the  {\constrained} problem.}
\end{figure*}
\begin{figure*}
\centering
\begin{minipage}[t]{0.24\linewidth}
    \includegraphics[width=\linewidth]{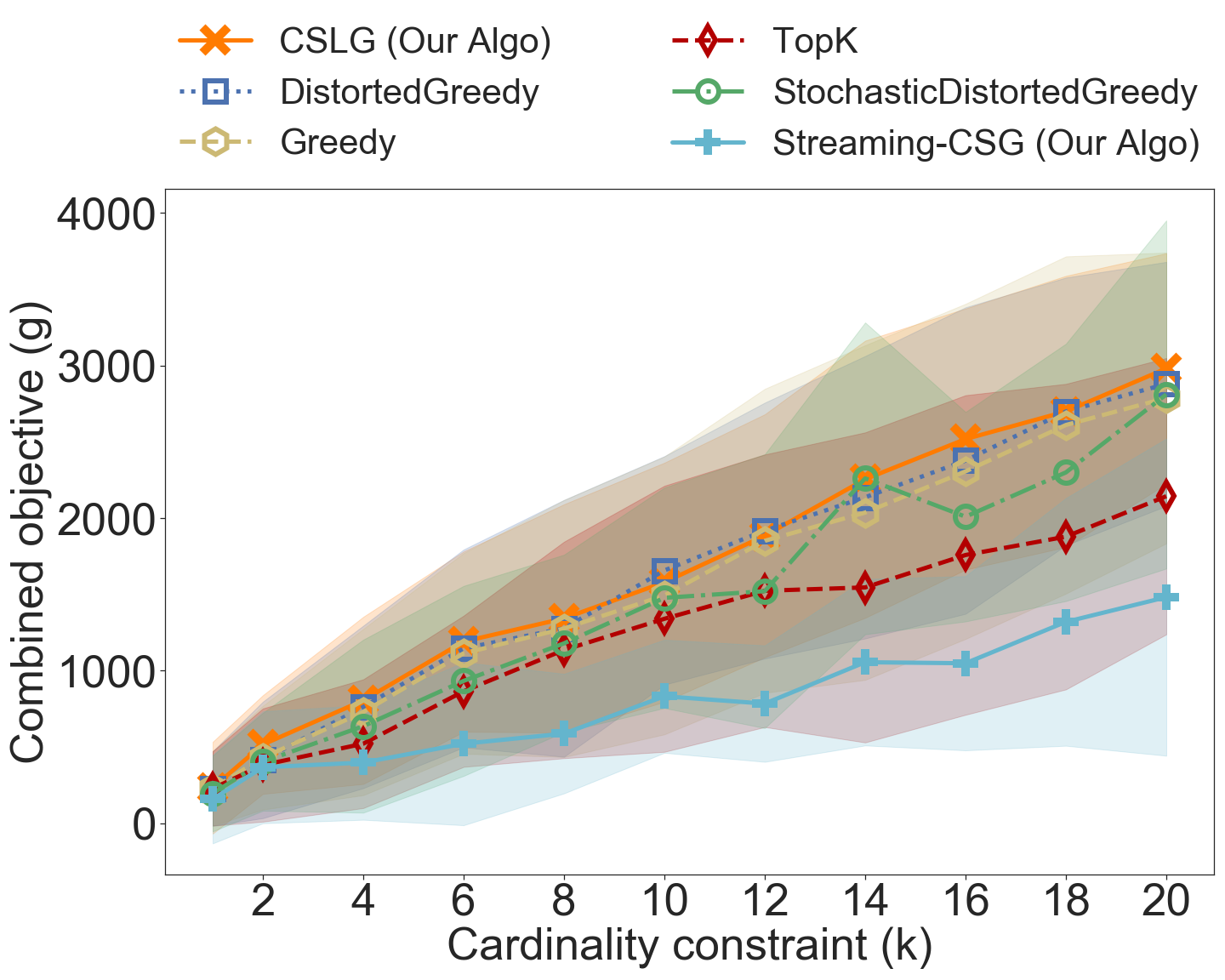}
    \subcaption{\label{fig:score_inf}\scriptsize{{\InfluenceDataset}}}
\end{minipage}
\begin{minipage}[t]{0.24\linewidth}
    \includegraphics[width=\linewidth]{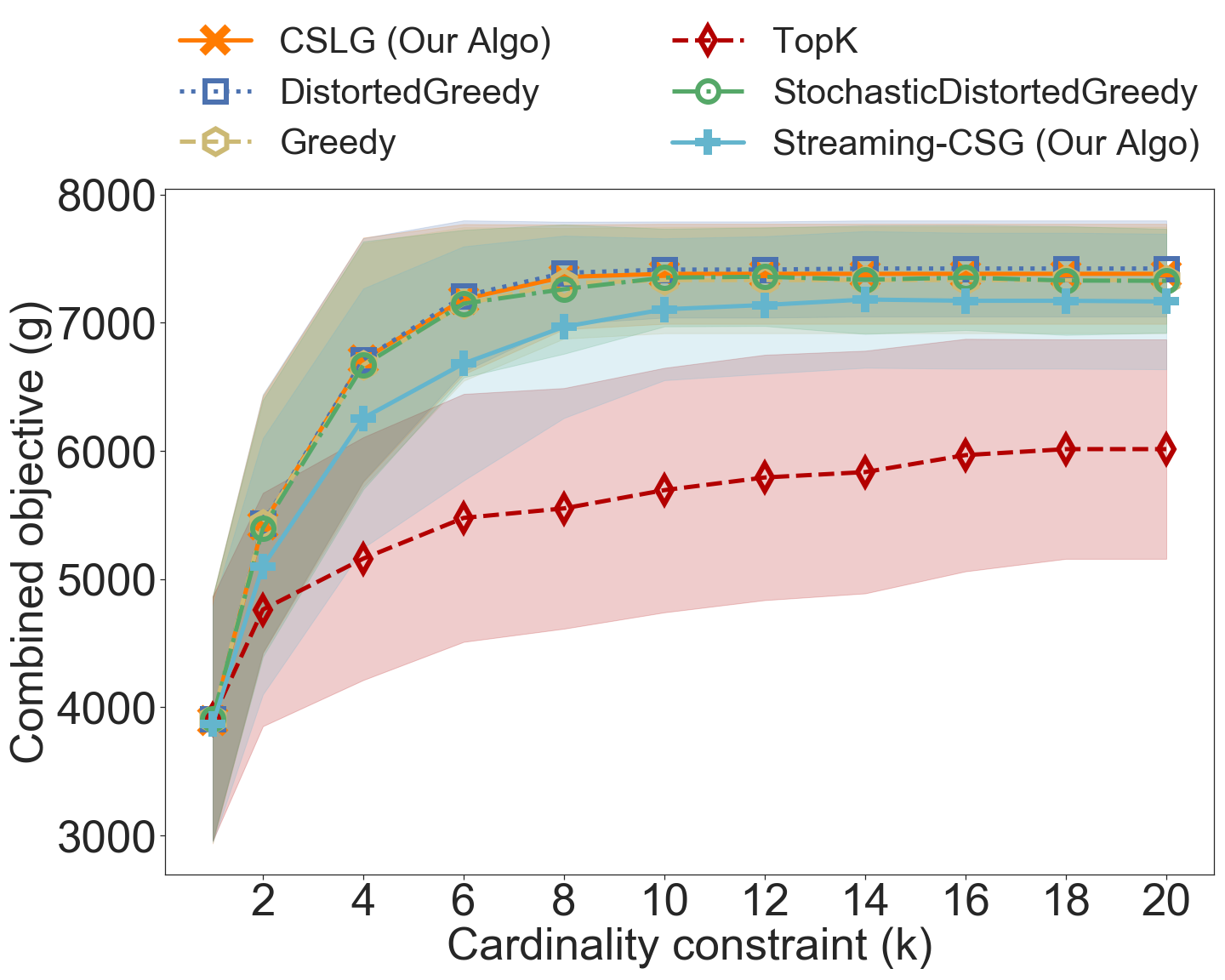}
    \subcaption{\label{fig:score_guru}\scriptsize{{\GuruDataset}}}
\end{minipage}
\begin{minipage}[t]{0.24\linewidth}
    \includegraphics[width=\linewidth]{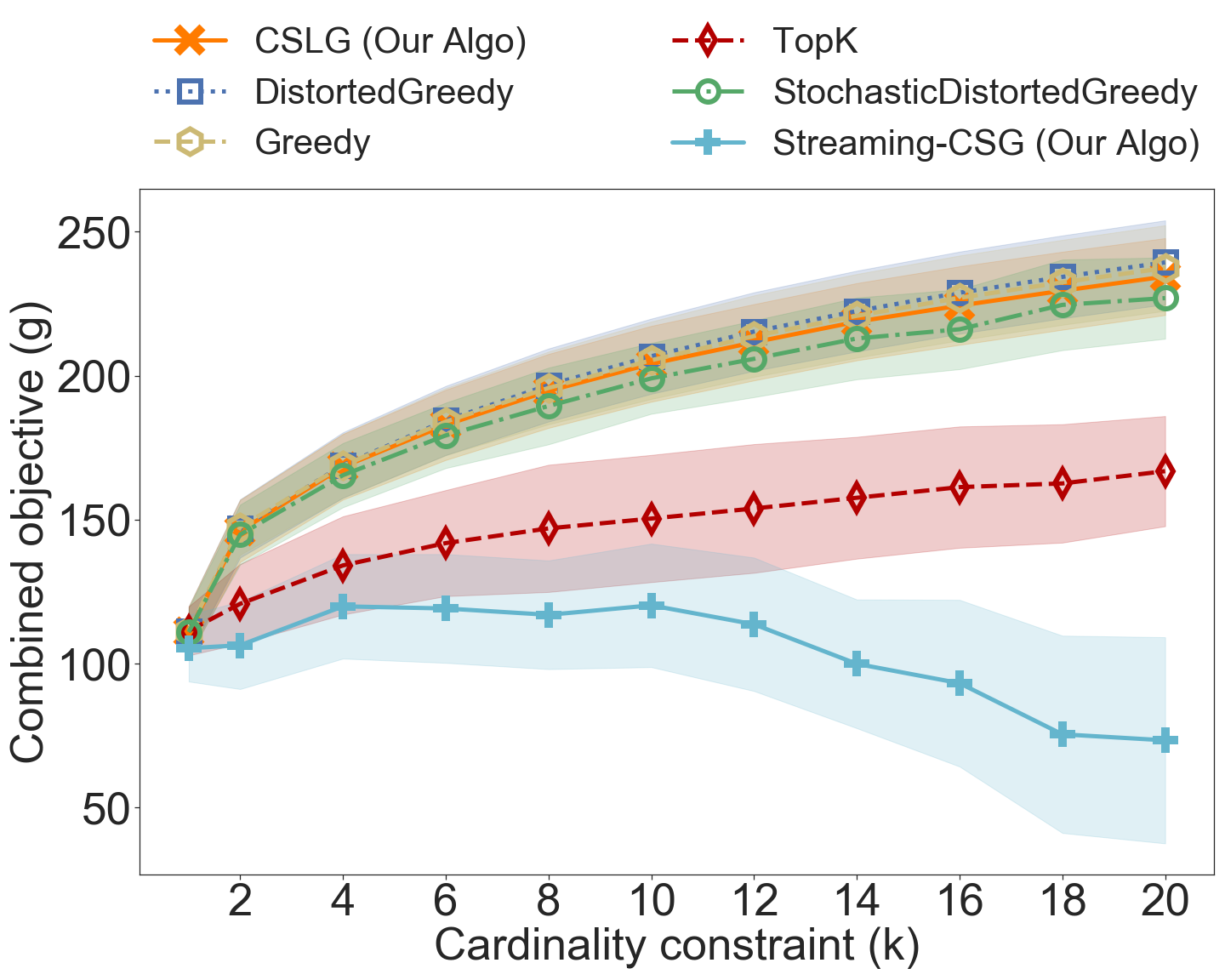}
    \subcaption{\label{fig:score_yelp}\scriptsize{{\YelpDataset}}}
\end{minipage}
\begin{minipage}[t]{0.24\linewidth}
    \includegraphics[width=\linewidth]{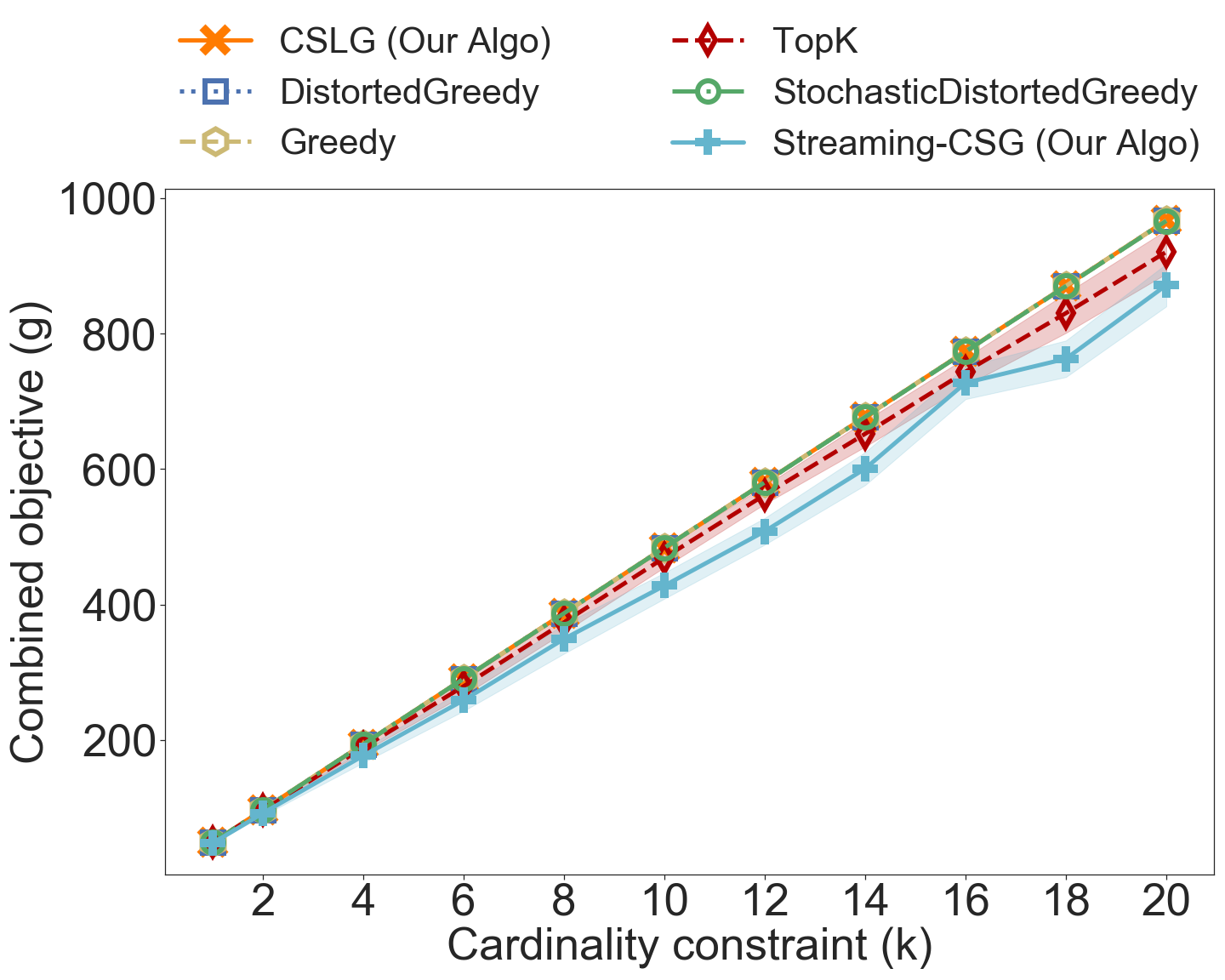}
    \subcaption{\label{fig:score_movie}\scriptsize{{\MovielensDataset}}}
\end{minipage}
\caption{\label{fig:perf_constrained} Combined objective value ($g$) comparisons of all algorithms for the {\constrained} problem.}
\end{figure*}

Here we evaluate the running time and empirical performance of the algorithms for the {\constrained} problem. 
In summary, we demonstrate that using our  algorithm {\costscaledlazygreedy}  can lead to faster running time without sacrificing the quality of the solution.

\spara{Runtime analysis:}
We start by evaluating the scalability of our methods. We vary the cardinality parameter $k$ and compute the running time of each algorithm.
The results of their running time performance are shown in Figure \ref{fig:time_constrained}.
The $y$-axis represents the running time (in sec), and the $x$-axis represents the cardinality $k$.

We note that {\distortedgreedy} and {\costscaledgreedy} require the most time to run and have very close performances.
Next, we consider the  {\stochasticdistortedgreedy} algorithm,  which is faster than the aforementioned algorithms because in each iteration it only evaluates the marginal gain of a subset of the elements.
We now draw the attention to the computational savings when using our proposed cost scaled greedy with lazy evaluations. 
In all datasets and specifically for larger values of $k$, a reasonable setting in all of our applications, {\costscaledlazygreedy} is more than 100x faster than both {\distortedgreedy} and {\costscaledgreedy}, and 10x faster than {\stochasticdistortedgreedy}. 
The only algorithm whose running time is comparable to {\costscaledlazygreedy} is {\topkexperts}, but the latter algorithm achieves lower objective value.

\spara{Performance evaluation:} 
Here we show that the aforementioned computational gains are achieved without sacrificing the solution quality.
For the performance evaluation we vary the cardinality parameter $k$ and compute the combined objective function 
($g$) of the obtained solution.
We present the results  in Figure \ref{fig:perf_constrained}.

We observe that the performance trends of the algorithms are overall consistent between all datasets and applications.
Note that as $k$ increases so does the objective value of the solutions found by the algorithms.
For the case of {\GuruDataset} and {\YelpDataset} we notice that the performance of the algorithms increases up until some point where it seems to stabilize.
A possible explanation is that initially the algorithms benefit from adding more elements to the solution because increasing the submodular gain outweighs the cost.
However, for some cardinality $k$ the algorithms may reach a solution where adding more elements does not benefit them. This happens in two cases; (i) when we have reached the maximum possible submodular value (e.g. covering all the requirements of a task), and (ii) when the benefit from increasing the submodular value is smaller than paying the corresponding cost.
We see that this is less pronounced in {\InfluenceDataset} and is not observed at all in {\MovielensDataset}.

Comparison across algorithms reveals that
the baseline {\topkexperts} and {\onlinekcostscaledgreedy} have the worse performance.
Intuitively,  the latter has a lower performance because it is an online algorithm and  we expect it to perform worse compared to the offline algorithms.
Among the offline algorithms, the {\stochasticdistortedgreedy} is slightly outperformed by {\distortedgreedy}, {\costscaledlazygreedy} and {\greedy} in all datasets except from {\MovielensDataset} where it has the same performance.
Finally, {\distortedgreedy}, {\costscaledlazygreedy} and {\greedy} perform similarly, with the last performing slightly worse for larger $k$ values. 
We note that even though {\greedy} is a heuristic  without provable approximation guarantees it still performs well. The comparison between {\distortedgreedy} and {\costscaledlazygreedy} shows that in practice the two algorithms perform the same for the cardinality constraint problem.
Overall, we see that {\costscaledlazygreedy} can achieve solutions of the same value as {\distortedgreedy} but is orders of magnitude faster as discussed above.

\subsection{Evaluation for {\large{\unconstrained}}}
\label{exp:unconstrained}

\begin{figure*}
\centering
\begin{minipage}[t]{0.24\linewidth} 
    \includegraphics[width=\linewidth]{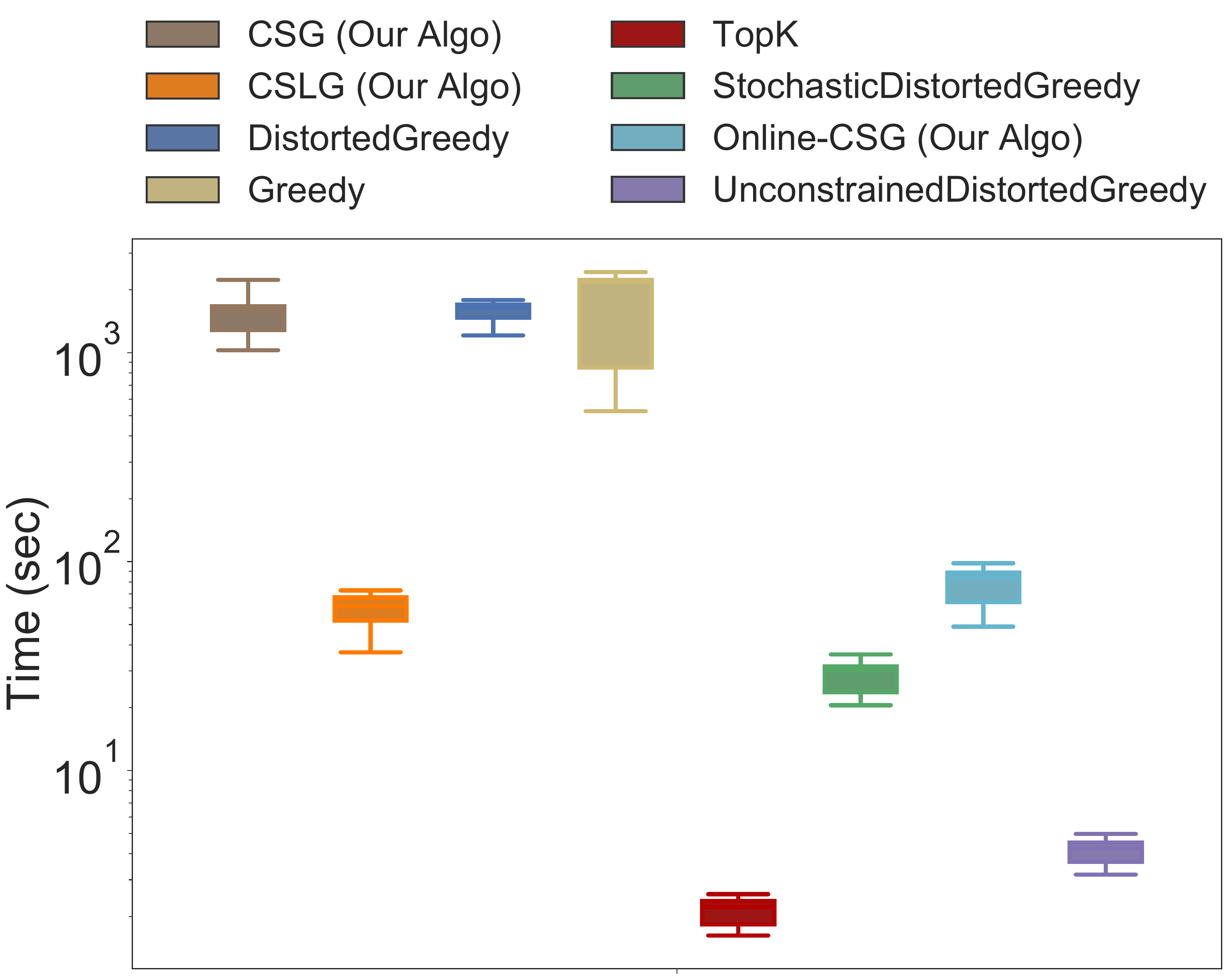}
    \subcaption{\label{fig:time_un_inf}\scriptsize{{\InfluenceDataset}}}
\end{minipage}
\begin{minipage}[t]{0.24\linewidth}
    \includegraphics[width=\linewidth]{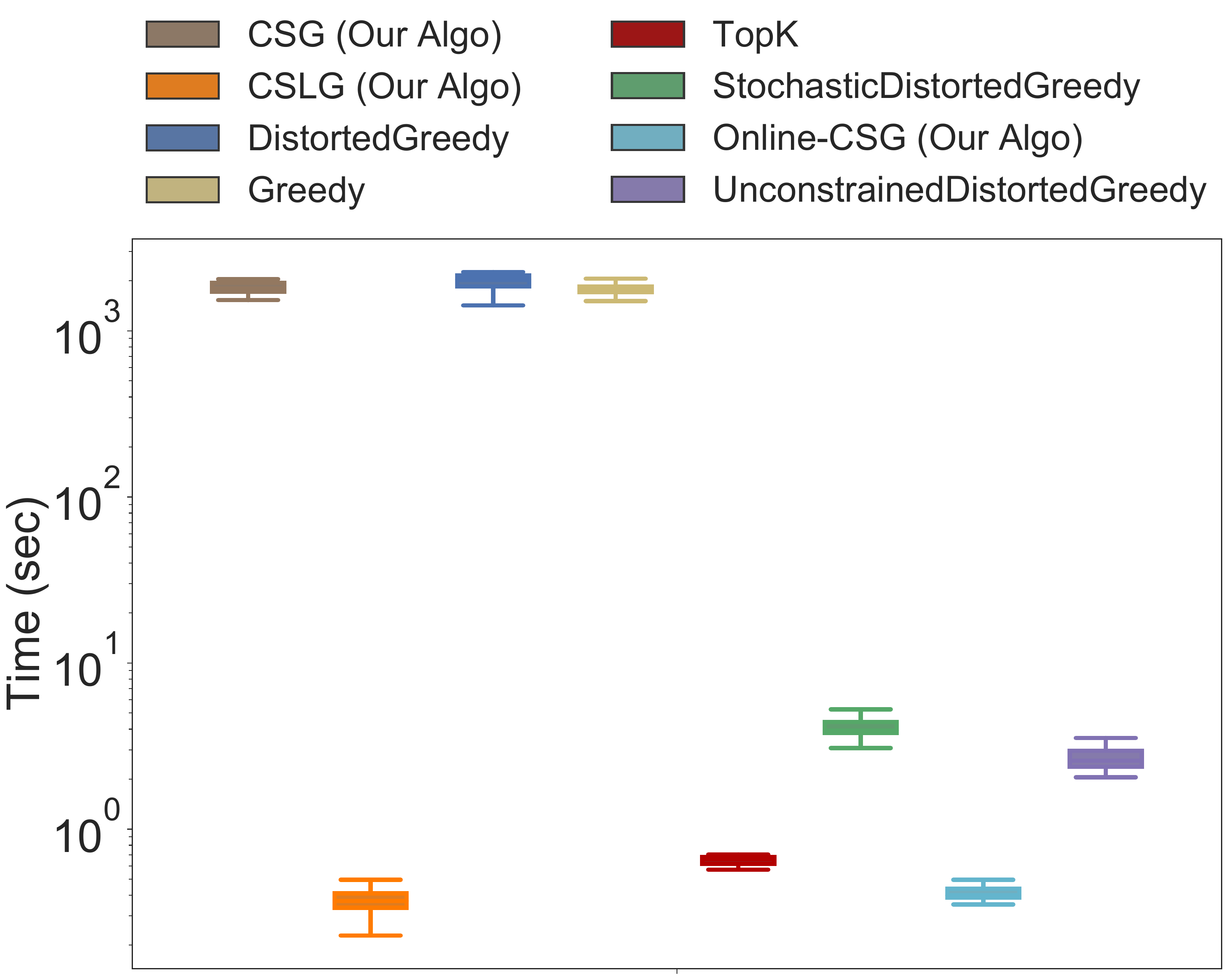}
    \subcaption{\label{fig:time_un_guru}\scriptsize{{\GuruDataset}}}
\end{minipage}
\begin{minipage}[t]{0.24\linewidth}
    \includegraphics[width=\linewidth]{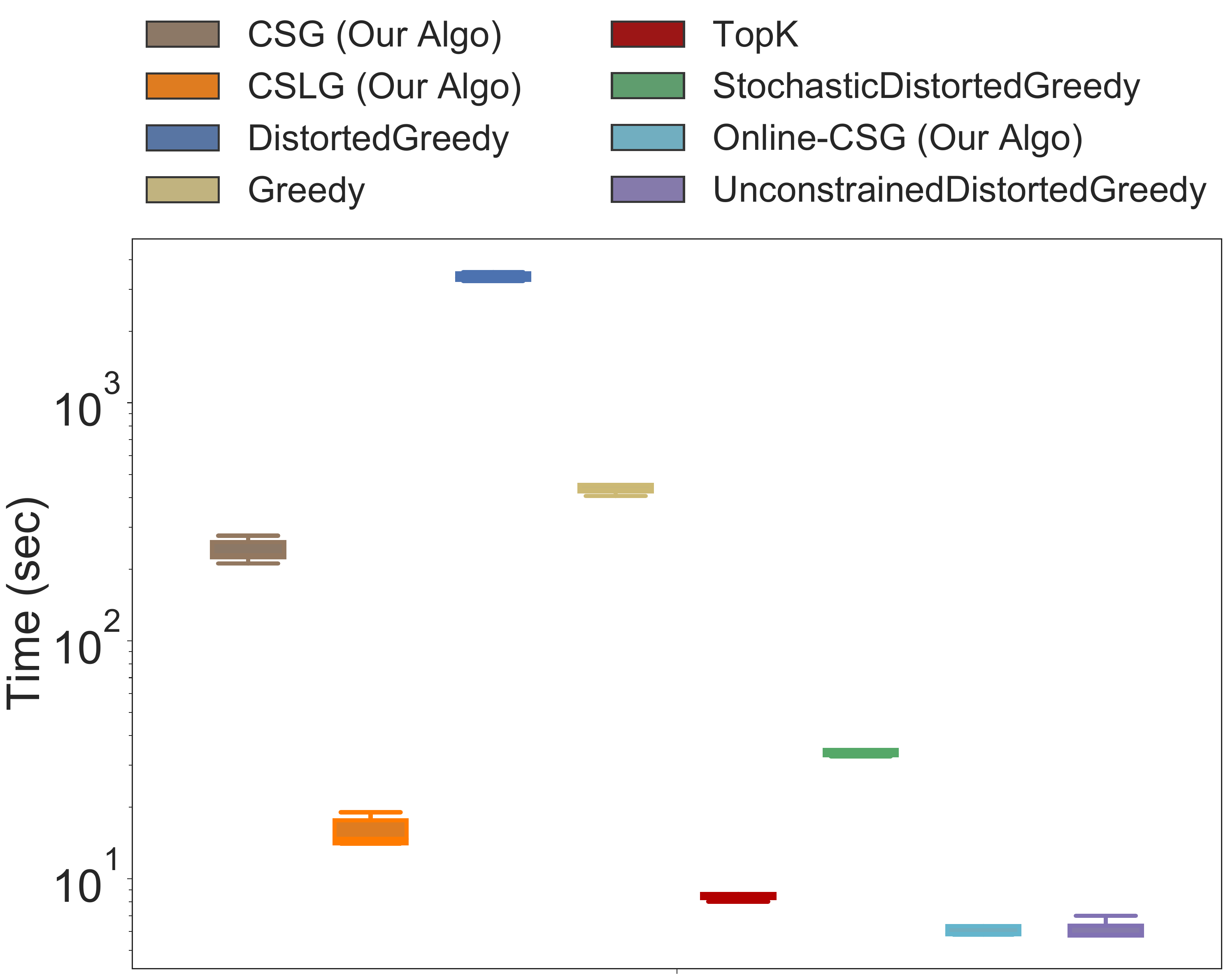}
    \subcaption{\label{fig:time_un_yelp}\scriptsize{{\YelpDataset}}}
\end{minipage}
\begin{minipage}[t]{0.24\linewidth}
    \includegraphics[width=\linewidth]{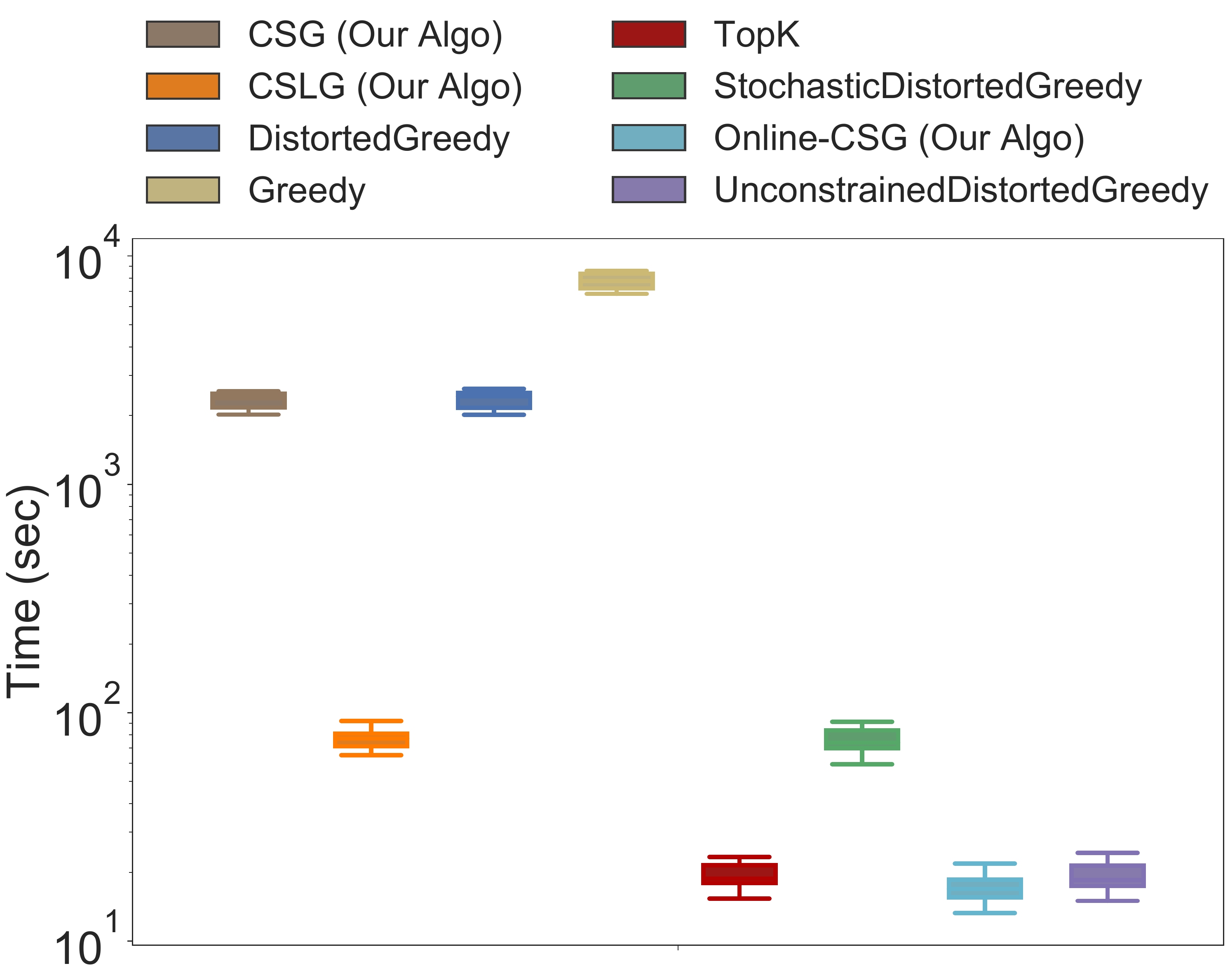}
    \subcaption{\label{fig:time_un_movie}\scriptsize{{\MovielensDataset}}}
\end{minipage}
\caption{\label{fig:time_unconstrained}  Running time (sec) comparisons of all algorithms for the {\unconstrained} problem.}
\end{figure*}
\begin{figure*}
\centering
\begin{minipage}[t]{0.24\linewidth}
    \includegraphics[width=\linewidth]{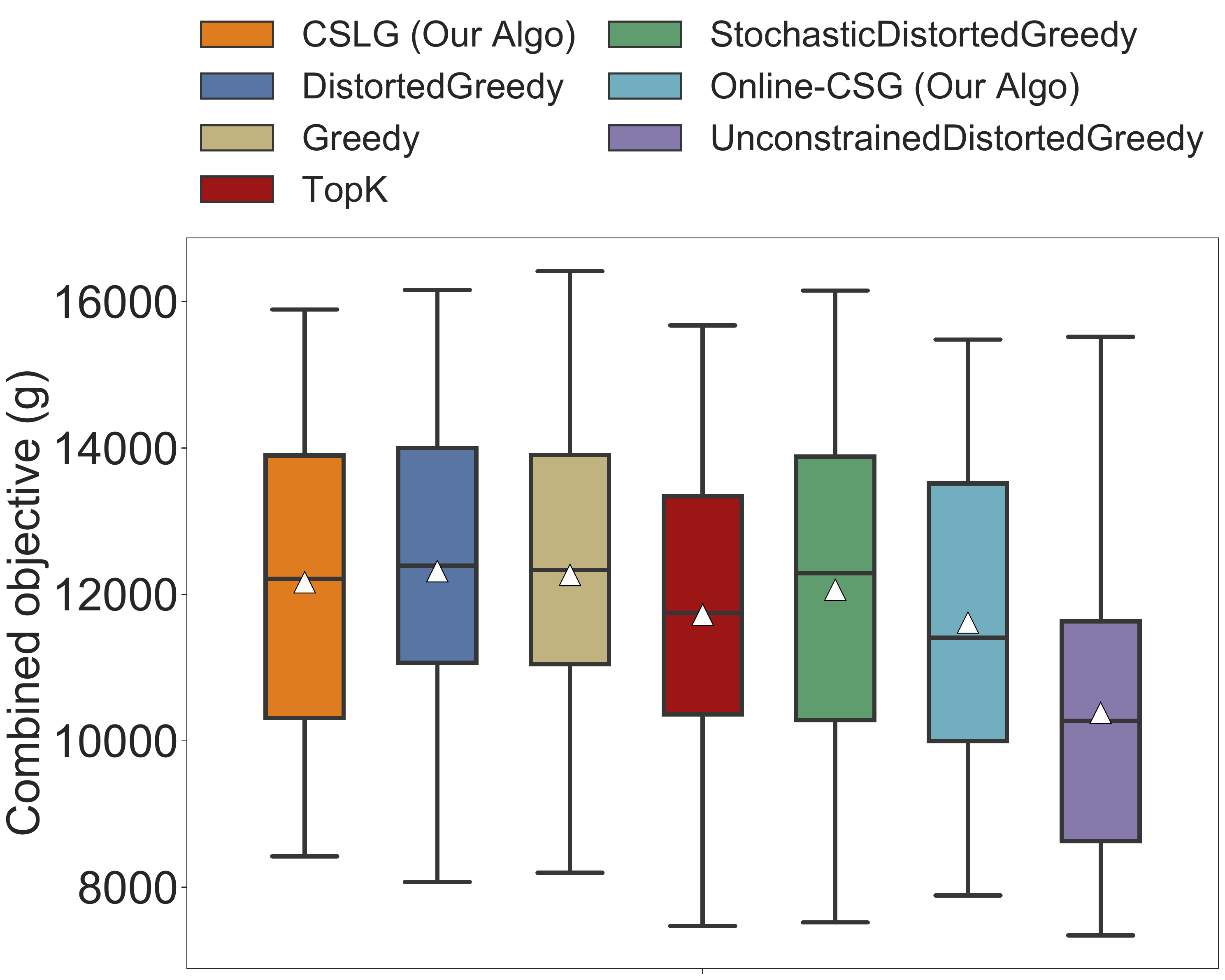}
    \subcaption{\label{fig:score_un_inf}\scriptsize{{\InfluenceDataset}}}
\end{minipage}
\begin{minipage}[t]{0.24\linewidth}
    \includegraphics[width=\linewidth]{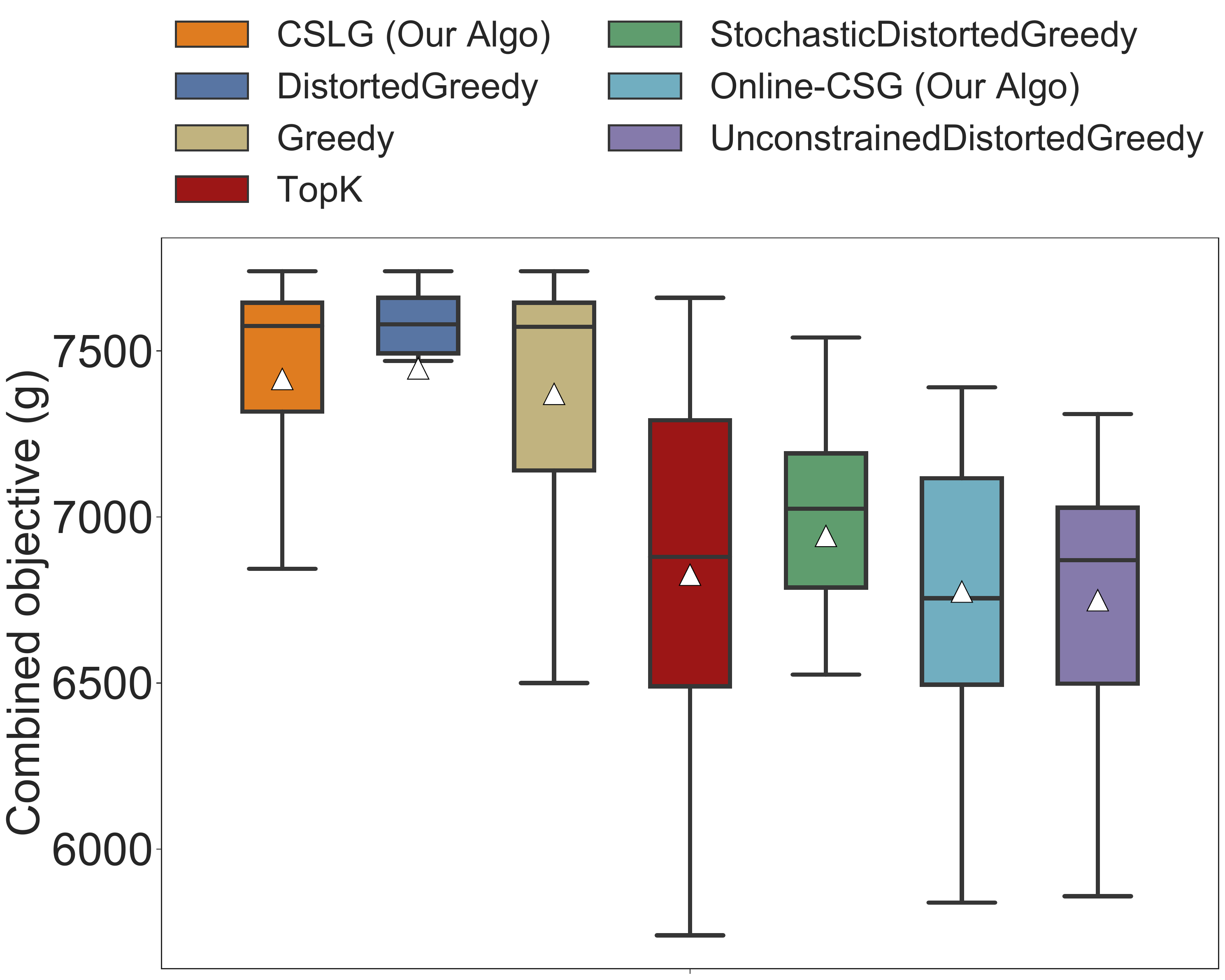}
    \subcaption{\label{fig:score_un_guru}\scriptsize{{\GuruDataset}}}
\end{minipage}
\begin{minipage}[t]{0.24\linewidth}
    \includegraphics[width=\linewidth]{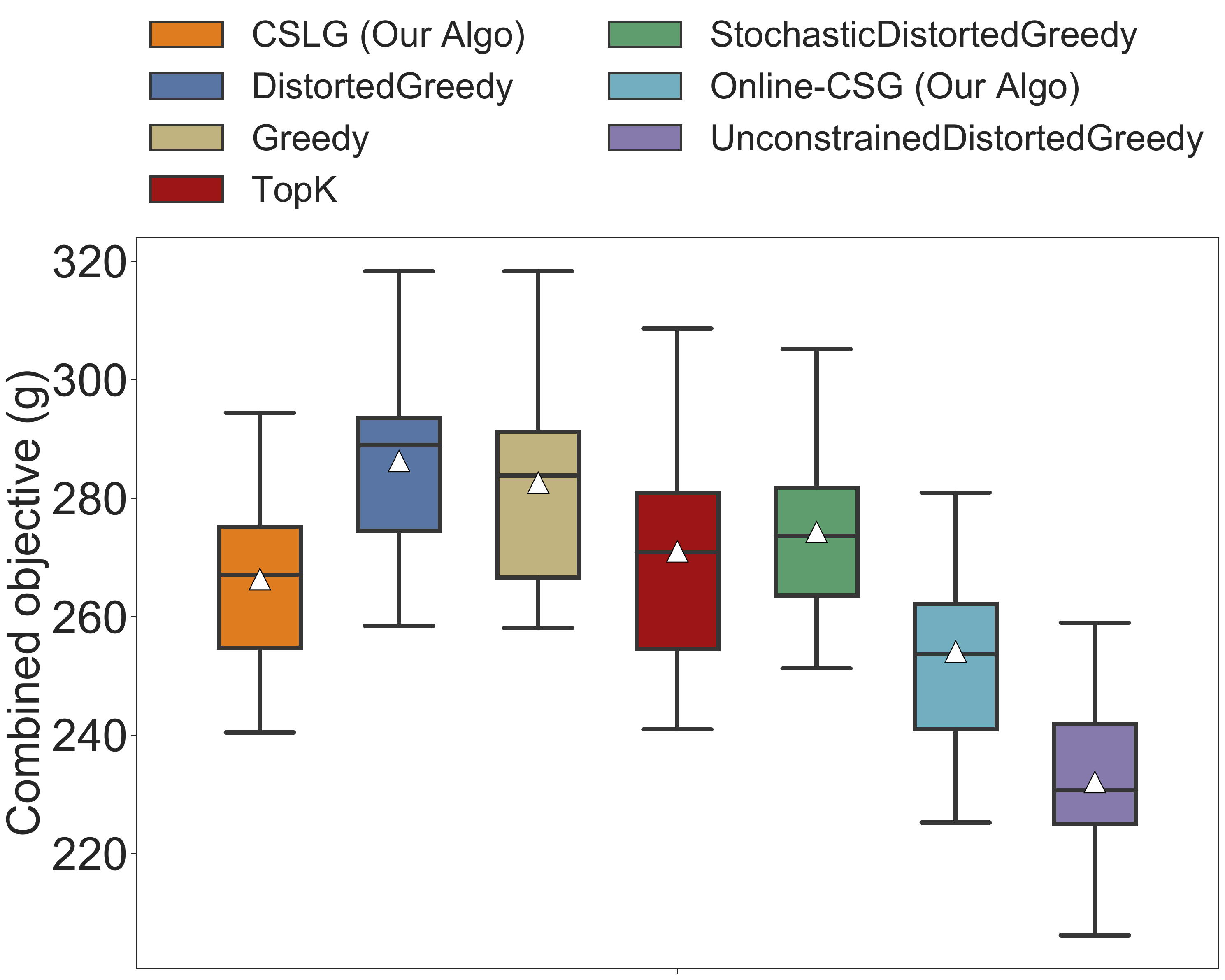}
    \subcaption{\label{fig:score_un_yelp}\scriptsize{{\YelpDataset}}}
\end{minipage}
\begin{minipage}[t]{0.24\linewidth}
    \includegraphics[width=\linewidth]{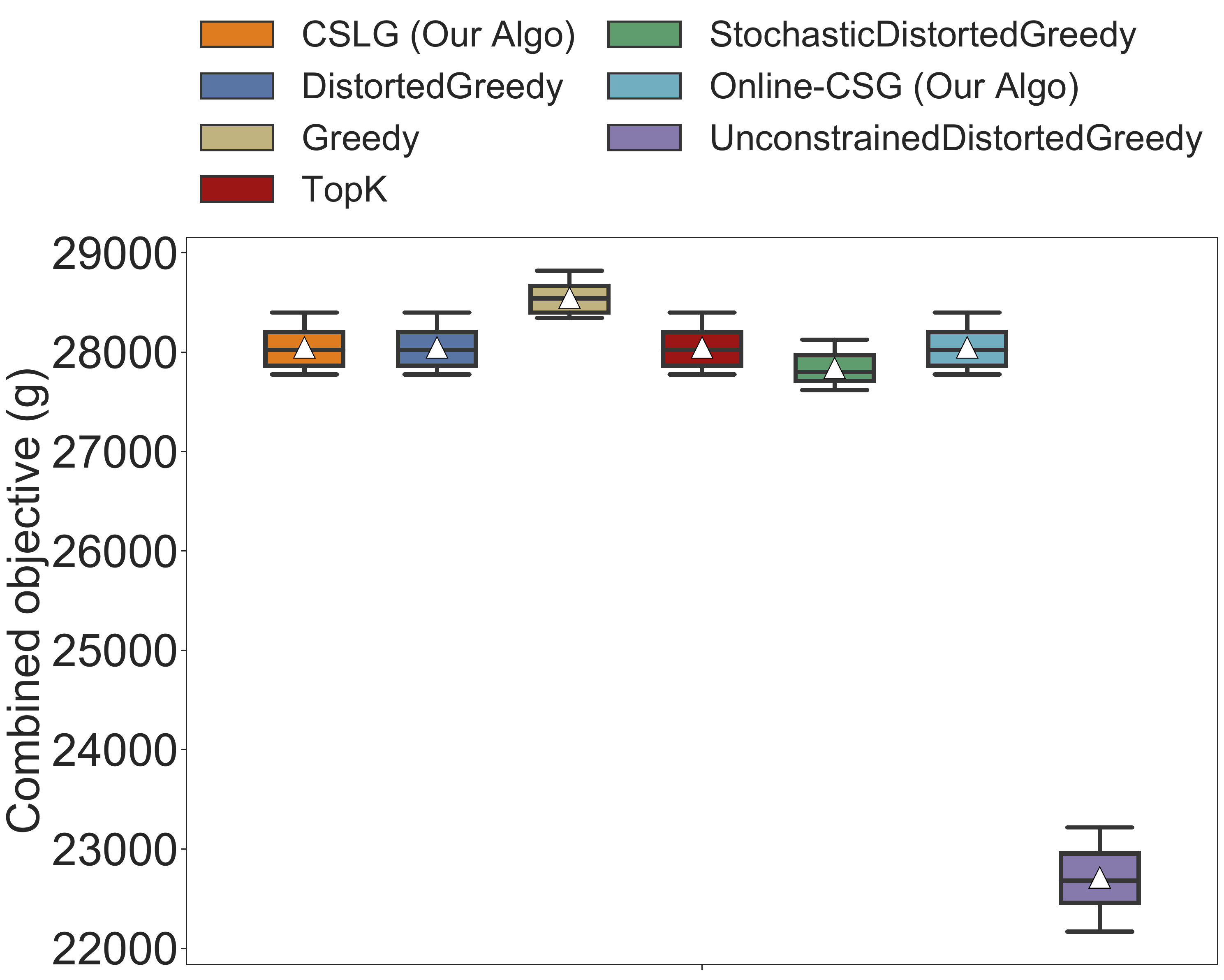}
    \subcaption{\label{fig:score_un_movie}\scriptsize{{\MovielensDataset}}}
\end{minipage}
\caption{\label{fig:perf_unconstrained} Combined objective value ($g$) comparisons of all algorithms for the {\unconstrained} problem.}
\end{figure*}
\begin{figure*}
\centering
\begin{minipage}[t]{0.24\linewidth}
    \includegraphics[width=\linewidth]{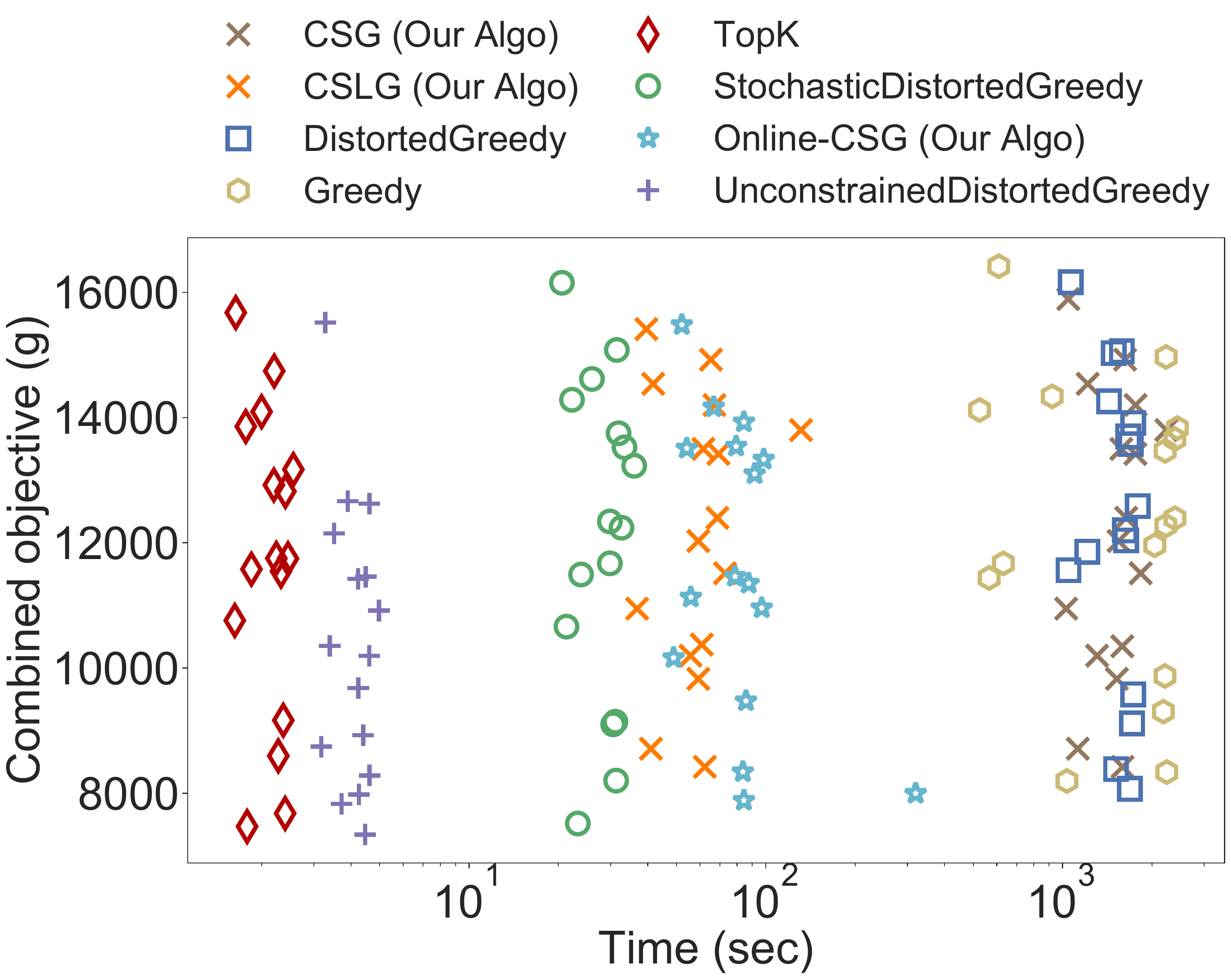}
    \subcaption{\label{fig:score_time_inf}\scriptsize{{\InfluenceDataset}}}
\end{minipage}
\begin{minipage}[t]{0.24\linewidth}
    \includegraphics[width=\linewidth]{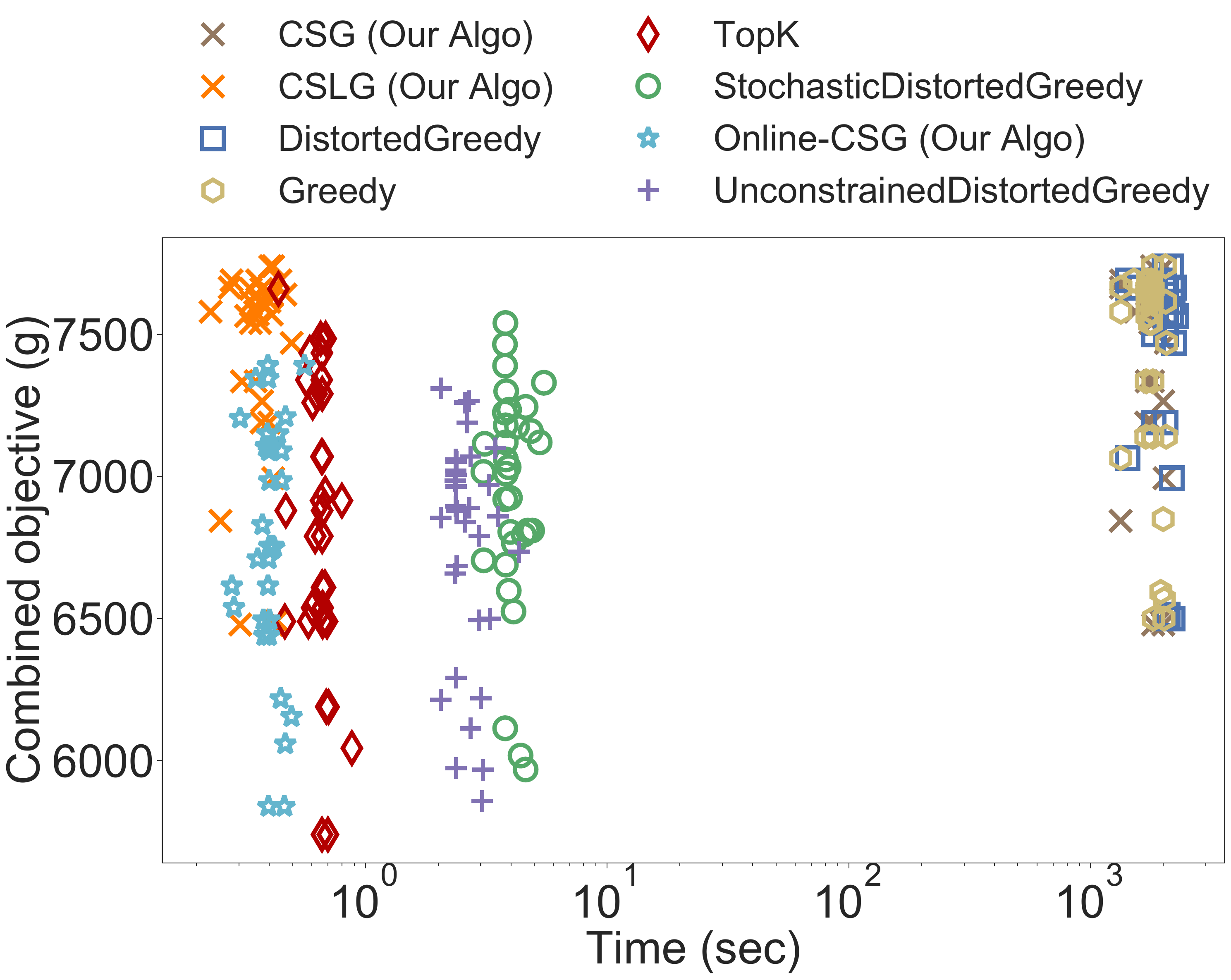}
    \subcaption{\label{fig:score_time_guru}\scriptsize{{\GuruDataset}}}
\end{minipage}
\begin{minipage}[t]{0.24\linewidth}
    \includegraphics[width=\linewidth]{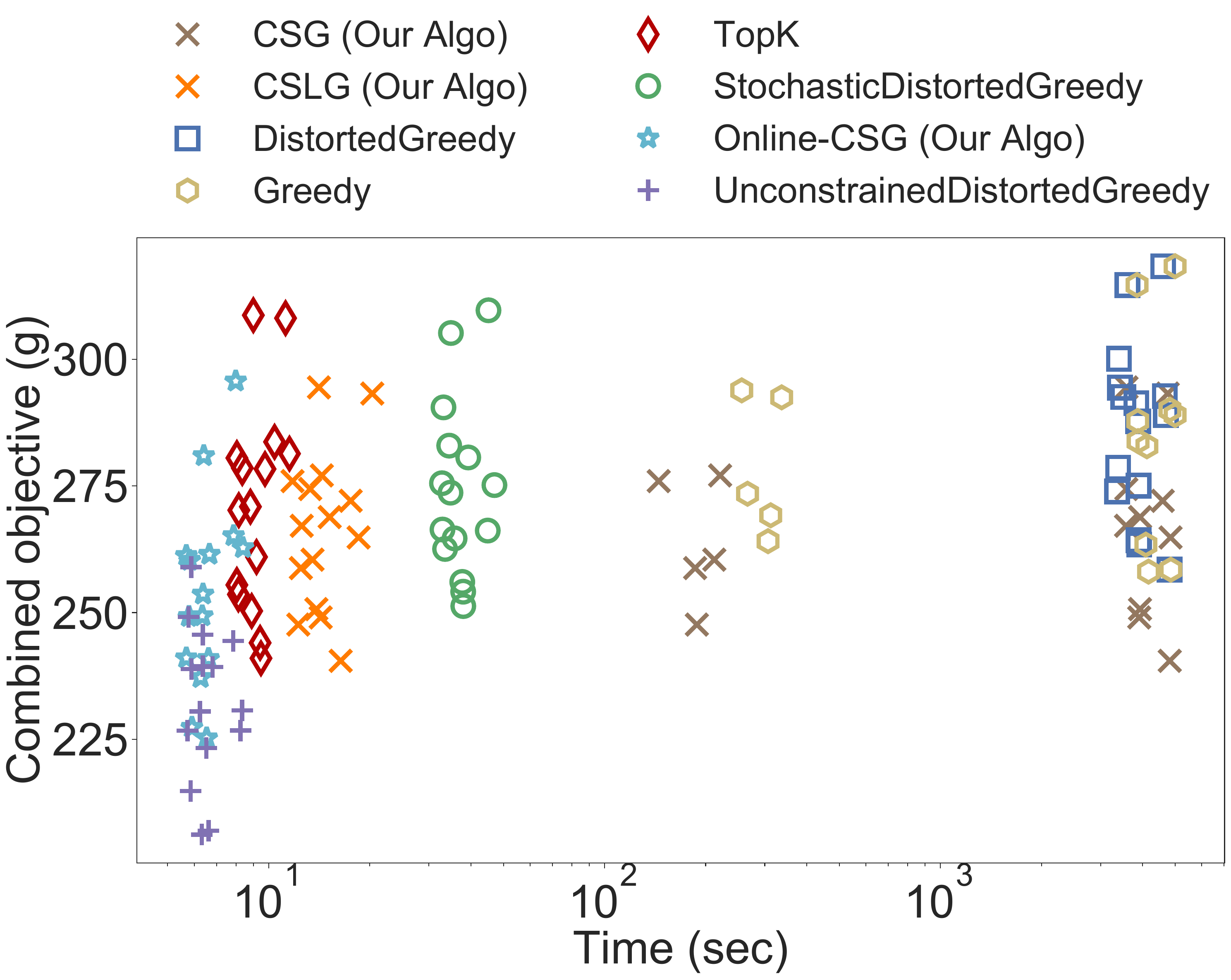}
    \subcaption{\label{fig:score_time_yelp}\scriptsize{{\YelpDataset}}}
\end{minipage}
\begin{minipage}[t]{0.24\linewidth}
    \includegraphics[width=\linewidth]{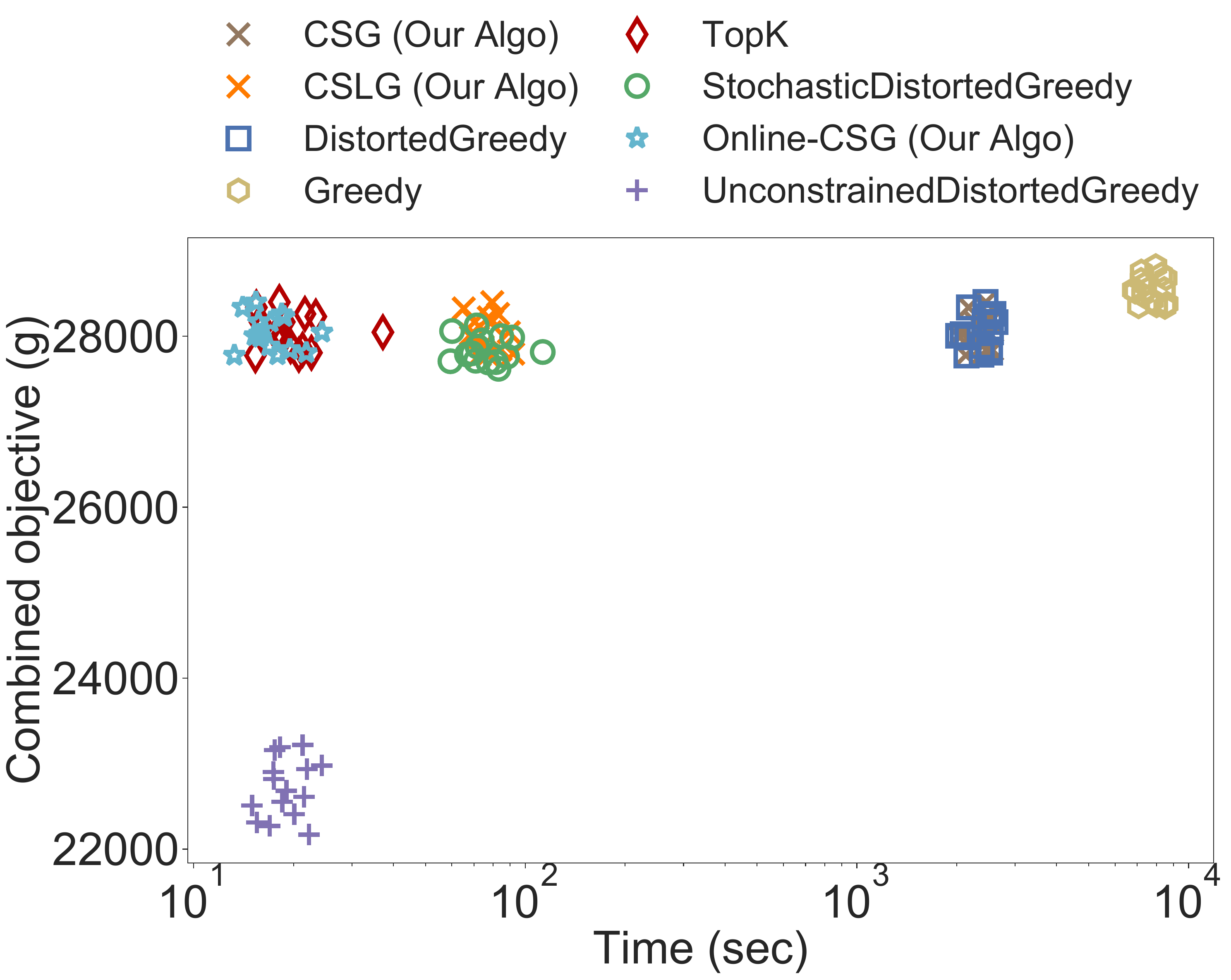}
    \subcaption{\label{fig:score_time_movie}\scriptsize{{\MovielensDataset}}}
\end{minipage}
\caption{\label{fig:perf_score_time_uncon} Combined objective value ($g$) and running time (sec) comparisons of all algorithms for the {\unconstrained} problem.}
\end{figure*}

Here we evaluate the algorithms for the {\unconstrained} problem.  
In this setting our algorithm can achieve even higher speedups at a slight performance cost.

\spara{Runtime analysis:}
We start by comparing the running times of each algorithm;
the results are shown in Figure \ref{fig:time_unconstrained}, where the $y$-axis is in the log-scale.
The box plots show the medians of the runtime performance 
of each algorithm and the short lengths of the box plots indicate small deviations from the mean; that is, the running time of the algorithms is consistent among all random samples.

We observe that {\costscaledgreedy} and {\distortedgreedy} have similar running times, with the former being slightly faster.
Overall however, their running times are orders of magnitude slower compared to the other algorithms.
Let us now investigate the gains we get by using lazy evaluations on {\costscaledgreedy}.
When considering our results at an application-wise level we see that the highest benefits are in the team-formation application (Figure~\ref{fig:time_un_guru}) where we see that {\costscaledlazygreedy} achieves 1000x of speedup, compared to the standard greedy based approach without lazy evaluations and to {\distortedgreedy}.
Slightly smaller gains are obtained in Figures~\ref{fig:time_un_inf},~\ref{fig:time_un_yelp} and~\ref{fig:time_un_movie} but they are still significant; i.e., in Figure~\ref{fig:time_un_inf} we see that while {\costscaledgreedy} and {\distortedgreedy} need half an hour to produce their results, {\costscaledlazygreedy} requires only 2 minutes.
In addition, we see that in these three cases a subset of the algorithms {\topkexperts}, {\onlinecostscaledgreedy}, {\distortedgreedy} and {\unconstraineddistortedgreedy} is faster than {\costscaledlazygreedy} but as we see next the performance of these algorithms with respect to the objective function is either similar to {\costscaledlazygreedy} or worse.

\vspace{1ex}
\noindent{\bf Performance evaluation:}
For the performance evaluation we compare the combined objective value ($g$) of the solution of each algorithm.
We present the results in Figure \ref{fig:perf_unconstrained}.
The box plots show medians (solid line), means (triangle) and interquartile ranges for the combined objective of each algorithm.

We note that {\costscaledlazygreedy} compares favorably to {\distortedgreedy}.  The performance of the two algorithms is similar in the team-formation (Figures~\ref{fig:score_un_guru}) and movie-recommendation (Figures~\ref{fig:score_un_movie}) applications. In the remaining applications (Figures~\ref{fig:score_un_inf} and~\ref{fig:score_un_yelp}), {\costscaledlazygreedy} has slightly worse performance but faster running time than  {\distortedgreedy}. Figure~\ref{fig:perf_score_time_uncon} summarizes the performance of all algorithms (across all iterations) for the two evaluation criteria: combined objective ($y$-axis) and running time ($x$-axis).

\section{Conclusions}\label{sec:conclusions}
In this paper, we focused on the problem of balancing between the goal of optimizing a submodular function by picking
a subset of elements from a collection with the actual cost of picking these elements. We formalized this problem
as the problem of optimizing a non-negative monotone submodular function minus a linear cost function, and  we designed effective and
efficient algorithms with provable approximation guarantees.  This framework enables us to generalize problem formulations that appear in many data-mining applications.
In our experiments we demonstrate that our proposed algorithms are highly efficient 
and, despite their slightly weaker theoretical bounds compared to existing work, their practical performance is equivalent to 
the latter in terms of the objective function.  
Finally, although we focused here on the cardinality-constraint version of the {\constrained} problem, we point out
that our results generalize to general matroid constraints, with significant practical and theoretical implications.

\spara{Acknowledgements:}  We acknowledge the support of NSF grants: III 1908510, III 1813406 and CAREER CCF-1750333.

\bibliographystyle{abbrv}
\bibliography{refs}

\clearpage
\appendix

\section{Analysis of Algorithm \ref{algo:cardinality}}
\label{app:offline}

In this section, we analyze our algorithm for the cardinality-constrained problem {\constrained}. We show the following guarantee.
\begin{thm}
Algorithm \ref{algo:cardinality} returns a solution $Q$ of size at most $k$ satisfying $\cov(Q)-\cost(Q)\geq\frac{1}{2}\cov(\opt)-\cost(\opt)$.
\end{thm}
\begin{proof}
The starting point of our analysis is the following ordering of the elements in  $Q \cup \opt$, 
which we call the Greedy ordering. The Greedy ordering orders the elements in $Q \cup \opt$ as
\begin{equation}
\label{eq:greedy-ordering}
 e_1, e_2, \dots, e_{|Q\cup \opt|} \tag{GreedyOrdering}
\end{equation}
where $e_i \in \arg\max_{e \in (Q \cup \opt) \setminus \{e_1, \dots, e_{i-1} \}} \tilde{g}(e_i \vert \{e_1, \dots, e_{i-1}\})$  for each $i \in [|Q \cup \opt|]$.  That is, we select the next element $e_i$ in the ordering to be the element from the remaining set with maximum marginal gain on top of the previously selected elements $e_1, \dots, e_{i-1}$.

It follows from the execution of the algorithm that the first $|Q|$ elements in the Greedy ordering (\ref{eq:greedy-ordering}) are the elements of $Q$ in the order in which they were added to $Q$ by the algorithm.

In the remainder of the analysis we first identify a solution, which is a prefix of the Greedy ordering, that we will analyze and show that its value is competitive with that of $\opt$. This solution is simply the first $\ell = |\opt|$ elements in the Greedy ordering, and we denote it by $S^{(\ell)}$. We analyze this solution and relate its value to $\opt$. We then relate the value of the solution $Q$ returned by the algorithm to the value of $S^{(\ell})$.

Let $S^{(i)} = \{e_1, \dots, e_i\}$ for all $1\leq i \leq |Q\cup \opt|$. Let $\ell = |\opt|$. As noted above, the solution $S^{(\ell)}$ plays a key role in our analysis.

{\bf Relating $S^{(\ell)}$ to $\opt$.} We now analyze the solution $S^{(\ell)}$ and relate it to $\opt$. To this end, we construct an appropriate mapping between $S^{(\ell)}$ and $\opt$ as follows.  Since $S^{(\ell)}$ and $\opt$ have the same size, there is a bijection $\pi: \opt \to S^{(\ell)}$ such that, for every $i \leq \ell$, $\pi^{-1}(e_i)$ appears after or at the same position as $e_i$ in the Greedy ordering (\ref{eq:greedy-ordering}), i.e., $\pi^{-1}(e_i) = e_j$ for some index $j \geq i$. We can obtain such a mapping $\pi$ by iteratively matching each element of $\opt$ to the earliest element of $S^{(\ell)}$ that is still unmatched. Since $|\opt| = |S^{(\ell)}|$ and $S^{(\ell)}$ is comprised of the first $\ell$ elements in the Greedy ordering, every element $o \in \opt$ will be matched to exactly one element $\pi(o) \in S^{(\ell)}$ such that $\pi(o)$ appears no later than $o$ in the Greedy ordering, as needed.

We can use this bijective mapping $\pi$ to ``charge'' $\opt$ to $S^{(\ell)}$ as follows. By construction of the Greedy ordering and $\pi$, for every $i \leq \ell$, we have
\begin{equation}
 \tilde{g}(e_i \vert S^{(i-1)}) \geq \tilde{g}(\pi^{-1}(e_i) \vert S^{(i-1)}) \label{eq1}
 \end{equation}
Let $\opt^{(i)} = \pi^{-1}(S^{(i)})$ for all $i \leq \ell$. By submodularity and the fact that $\opt^{(i)} = \opt^{(i-1)} \cup \{\pi^{-1}(e_i)\}$, we have
\begin{align}
 \tilde{g}(\pi^{-1}(e_i) \vert S^{(i-1)}) 
&\geq \tilde{g}(\pi^{-1}(e_i) \vert S^{(\ell)} \cup \opt^{(i-1)}) \nonumber\\
& = \tilde{g}(S^{(\ell)}\cup \opt^{(i-1)}) - \tilde{g}(S^{(\ell)} \cup \opt^{(i-1)}) \label{eq2}
\end{align}
By combining (\ref{eq1}) and (\ref{eq2}), we obtain
\[ \tilde{g}(e_i \vert S^{(i-1)}) \geq \tilde{g}(S^{(\ell)}\cup \opt^{(i-1)}) - \tilde{g}(S^{(\ell)} \cup \opt^{(i-1)})\]
We sum up the above inequalities over all $i \leq \ell$. Note that the sums telescope. Additionally, we have $\opt^{(\ell)} = \pi^{-1}(S^{(\ell)}) = \opt$. Thus we obtain
\[ \tilde{g}(S^{(\ell)}) - \tilde{g}(\emptyset) \geq \tilde{g}(S^{(\ell)} \cup \opt) - \tilde{g}(S^{(\ell)}) \] 
and thus
\begin{equation}
\label{eq:Sell}
\tilde{g}(S^{(\ell)}) \geq \frac{1}{2} \tilde{g}(S^{(\ell)} \cup \opt) 
\end{equation}

{\bf Relating $Q$ to $S^{(\ell)}$.} We now relate the solution $Q$ constructed by the algorithm to the solution $S^{(\ell)}$. Recall that it follows from the execution of the algorithm that $Q$ is a prefix of the Greedy ordering. By definition, $S^{(\ell)}$ is also a prefix of the Greedy ordering. However, $Q$ and $S^{(\ell)}$ may be different prefixes and one may be included in the other, and we consider each of these cases in turn. To relate their values, we crucially use the following properties ensured by the algorithm: each element of $Q$ has positive marginal gain with respect to the scaled objective $\tilde{g}$ on top of the elements that come before it in the Greedy ordering; additionally, if $Q$ has less than $k$ elements, all of the remaining elements have non-positive marginal gain with respect to $\tilde{g}$ on top of $Q$. These properties follow from the fact that, when each element is added to $Q$, it has positive marginal gain with respect to $\tilde{g}$. Moreover, the algorithm terminates when either it reaches the size constraint $k$ or it terminates early on line~\ref{line:terminate} since the marginal gains of the remaining elements are non-positive with respect $\tilde{g}$.

We now give the precise  analysis.  We will show that $\tilde{g}(Q) \geq \tilde{g}(S^{(\ell)})$ by considering two cases: $|Q| \geq \ell$ and $|Q| < \ell$. 

Suppose $|Q| \geq \ell$. We have $S^{(\ell)} \subseteq Q$. Since the algorithm only adds elements with positive marginal gain, we have
\[ \tilde{g}(Q) - \tilde{g}(S^{(\ell)}) = \sum_{i=\ell+1}^{|Q|} \tilde{g}(e_i \vert S^{(i-1)}) \geq 0 \]
Suppose $|Q| < \ell$. We have $Q \subseteq S^{(\ell)}$. Since the algorithm terminates when the marginal gain of every element becomes non-positive, we have
\[ \tilde{g}(S^{(\ell)}) - \tilde{g}(Q) = \sum_{i = |Q|+1}^{\ell} \tilde{g}(e_i \vert S^{(i-1)}) \leq  \sum_{i = |Q|+1}^{\ell} \tilde{g}(e_i \vert Q) \leq 0 \]
Thus, in either case, we have that
\begin{equation}
\label{eq:scaled-obj}
\tilde{g}(Q) \geq \tilde{g}(S^{(\ell)})
\end{equation}
{\bf Relating $Q$ to $\opt$.}  We now put everything together and establish the approximation guarantee. By (\ref{eq:Sell}) and (\ref{eq:scaled-obj}), we have
\[ \tilde{g}(Q) \geq  \frac{1}{2} \tilde{g}(S^{(\ell)} \cup \opt) \]
Recall that $\tilde{g} = \cov - 2\cost$. Thus
\[ \cov(Q) - \cost(Q) \geq  \frac{1}{2} \cov(S^{(\ell)} \cup \opt) - \cost(S^{(\ell)} \cup \opt) + \cost(Q) \]
Since $\cov$ is monotone, we have $\cov(S^{(\ell)} \cup \opt) \geq \cov(\opt)$. Thus
\[ \cov(Q) - \cost(Q) \geq  \frac{1}{2} \cov(\opt) - \cost(S^{(\ell)} \cup \opt) + \cost(Q) \]
Thus, to finish the proof, it only remains to verify that
\[ \cost(\opt) + \cost(Q) \geq \cost(S^{(\ell)} \cup \opt) \]
As before, we consider two cases: $|Q| \geq \ell$ and $|Q| < \ell$. Suppose that $|Q| \geq \ell$. Then $S^{(\ell)} \subseteq Q$ and thus $\cost(Q) \geq \cost(S^{(\ell)})$, since $\cost$ is non-negative. Thus
\[ \cost(\opt) + \cost(Q) \geq \cost(\opt) + \cost(S^{(\ell)}) \geq \cost(\opt \cup S^{(\ell)}) \]
Suppose that $|Q| < \ell$. Then $Q \subseteq S^{(\ell)}$ and $S^{(\ell)} \setminus Q \subseteq \opt$. Thus $\opt \cup Q = \opt \cup S^{(\ell)}$ and hence
\[ \cost(\opt) + \cost(Q) \geq \cost(\opt \cup Q) = \cost(\opt \cup S^{(\ell)}) \]
Putting everything together, we have
\[ \cov(Q) - \cost(Q) \geq \frac{1}{2} \cov(\opt) - \cost(\opt) \]
\end{proof}

\section{Analysis of Algorithm~\ref{algo:online}}
\label{app:online}

In this section, we analyze our online algorithm for the  {\unconstrained}  problem. We show the following guarantee:
\begin{thm}
Algorithm~\ref{algo:online} returns a solution $Q$ satisfying $\cov(Q)-\cost(Q)\geq\frac{1}{2}\cov(\opt)-\cost(\opt)$.
\end{thm}
\begin{proof}
For every item $o\in\opt\setminus Q$, $\tilde{g}(o\vert Q)\leq0$ holds. This is due to the fact that $o$ had non-positive marginal gain when it arrived and the marginal gains can only decrease, since $\tilde{g}$ is submodular. Therefore
\begin{align*}
0 & \geq\sum_{o\in\opt\setminus Q}\tilde{g}(o\vert Q)\\
 & \geq\tilde{g}(Q\cup\opt)-\tilde{g}(Q)\\
 & =\Big(\underbrace{\cov(Q\cup\opt)}_{\geq \cov(\opt)}-\cov(Q)\Big)-2\Big(\underbrace{\cost(Q\cup\opt)-\cost(Q)}_{=\cost(\opt\setminus Q)\leq \cost(\opt)}\Big)\\
 & \geq \cov(\opt)-\cov(Q)-2\cost(\opt)
\end{align*}

The third inequality is by monotonicity of $\cov$ and non-negativity and linearity of $\cost$. The second inequality follows from submodularity as follows. Let $O=\opt\setminus Q$ and let $o_{1},o_{2},\dots,o_{|O|}$ be an arbitrary ordering of $O$. Let $O^{(i)}=\{o_{1},\dots,o_{i}\}$.
Then,
\begin{align*}
\tilde{g}(Q\cup O)-\tilde{g}(Q)
&=\sum_{i=1}^{|O|}\left(\tilde{g}(Q\cup O^{(i)})-\tilde{g}(Q\cup O^{(i-1)})\right)\\
=\sum_{i=1}^{|O|}\tilde{g}(o_{i}\vert Q\cup O^{(i-1)})
& \leq\sum_{i=1}^{|O|}\tilde{g}(o_{i}\vert Q)
\end{align*}
where the inequality is by submodularity.

Rearranging, we obtain
\[
\cov(Q)\geq \cov(\opt)-2\cost(\opt)
\]
On the other hand, since the algorithm only added elements with positive marginal gain with respect to $\tilde{g}$, we have
\[
\tilde{g}(Q)>0
\]
Indeed, let $e_{1},e_{2},\dots,e_{|Q|}$ be the elements of $Q$ in the order in which they were added. Let $Q^{(i)}=\{e_{1},\dots,e_{i}\}$. We have
\begin{align*}
\tilde{g}(Q)-\tilde{g}(\emptyset)
&=\sum_{i=1}^{|Q|}\left(\tilde{g}(Q^{(i)})-\tilde{g}(Q^{(i-1)})\right)\\
&=\sum_{i=1}^{|Q|}\tilde{g}(e_{i}\vert Q^{(i-1)})
>0
\end{align*}

 Since $\tilde{g}(\emptyset)=\cov(\emptyset)-\cost(\emptyset)=\cov(\emptyset)\geq0$,
we have $\tilde{g}(Q)>0$. Therefore,
\[
\cov(Q)-2\cost(Q)>0\Rightarrow \cost(Q)<\frac{1}{2}\cov(Q)\Rightarrow \cov(Q)-\cost(Q)>\frac{1}{2}\cov(Q)
\]
By combining with the previous inequality, we obtain
\begin{align*}
\cov(Q)-\cost(Q)
&> \frac{1}{2}\cov(Q)\\
&\geq\frac{1}{2}\left(\cov(\opt)-2\cost(\opt)\right)\\
&=\frac{1}{2}\cov(\opt)-\cost(\opt)
\end{align*}
\end{proof}

\section{Analysis of Algorithm~\ref{algo:streaming}}
\label{app:streaming}

In this section, we analyze our streaming algorithm for the  {\constrained} problem. 
The following theorem assumes that the parameters can be set
appropriately if we know the value of the optimal solution. This assumption can be removed using a technique due to  \cite{badanidiyuru2014streaming}. 

\begin{thm}
When run with scaling constant $\const=\frac{1}{2}\left(3+\sqrt{5}\right)$ and threshold $\tau=\frac{1}{k}\left(\frac{1}{2}(3-\sqrt{5})\cov(\opt)-\cost(\opt)\right)$, Algorithm \ref{algo:streaming} returns a solution $Q$ such that $|Q| \leq k$ and
\[ \cov(Q)-\cost(Q)\geq \frac{1}{2}\left(3-\sqrt{5}\right)\cov(\opt)-\cost(\opt) \]
\end{thm}
\begin{proof}
It is clear from the execution of the algorithm that $|Q| \leq k$. Therefore we focus on analyzing the function value. We consider two cases, depending on whether $|Q|=k$ or $|Q|<k$.

\spara{Case 1: $|Q|=k$.}
We have
\[
\tilde{g}(Q)\geq\tau k\Rightarrow \cov(Q)-\const\cdot \cost(Q)\geq\tau k
\]
\spara{Case 2: $|Q|<k$.}
For every item $o\in\opt\setminus Q$, we have
$
\tilde{g}(o\vert Q)\leq\tau
$. 
This is due to the fact that $o$ had marginal gain less than $\tau$ when it arrived and the marginal gains can only decrease due to submodularity of $\tilde{g}$. Therefore
\begin{align*}
\tau|\opt\setminus Q| & \geq\sum_{o\in\opt\setminus Q}\tilde{g}(o\vert Q)\\
 & \geq\tilde{g}(Q\cup\opt)-\tilde{g}(Q)\\
 & =\Big(\underbrace{\cov(Q\cup\opt)}_{\geq \cov(\opt)}-\cov(Q)\Big)-\const\Big(\underbrace{\cost(Q\cup\opt)-\cost(Q)}_{=\cost(\opt\setminus Q)\leq \cost(\opt)}\Big)\\
 & \geq \cov(\opt)-\cov(Q)-\const\cdot \cost(\opt)
\end{align*}
The third inequality is by monotonicity of $\cov$ and non-negativity and linearity of $c$. The second inequality follows from submodularity as follows. Let $O=\opt\setminus Q$ and let $o_{1},o_{2},\dots,o_{|O|}$ be an arbitrary ordering of $O$. Let $O^{(i)}=\{o_{1},\dots,o_{i}\}$. Then
\begin{align*}
\tilde{g}(Q\cup O)-\tilde{g}(Q)
&=\sum_{i=1}^{|O|}\left(\tilde{g}(Q\cup O^{(i)})-\tilde{g}(Q\cup O^{(i-1)})\right)\\
=\sum_{i=1}^{|O|}\tilde{g}(o_{i}\vert Q\cup O^{(i-1)})
&\leq\sum_{i=1}^{|O|}\tilde{g}(o_{i}\vert Q)
\end{align*}
where the inequality is by submodularity.

Rearranging, we obtain
\begin{align*}
\cov(Q)
&\geq \cov(\opt)-\const\cdot \cost(\opt)-\tau\underbrace{|\opt\setminus Q|}_{\leq k}\\
&\geq \cov(\opt)-\const\cdot \cost(\opt)-\tau k
\end{align*}
On the other hand, since the algorithm only added elements with marginal gain at least the threshold, we can show that
\[
\tilde{g}(Q)\geq\tau|Q|
\]
Indeed, let $e_{1},e_{2},\dots,e_{|Q|}$ be the elements of $Q$ in the order in which they were added. Let $Q^{(i)}=\{e_{1},\dots,e_{i}\}$. We have
\begin{align*}
\tilde{g}(Q)-\tilde{g}(\emptyset)
=\sum_{i=1}^{|Q|}\left(\tilde{g}(Q^{(i)})-\tilde{g}(Q^{(i-1)})\right)
=\sum_{i=1}^{|Q|}\tilde{g}(e_{i}\vert Q^{(i-1)})\geq\tau|Q|
\end{align*}
Since $\tilde{g}(\emptyset)=\cov(\emptyset)-\cost(\emptyset)=\cov(\emptyset)\geq0$,
we have $\tilde{g}(Q)\geq\tau|Q|$. Thus
\[ \cov(Q)-\const\cdot \cost(Q)\geq\tau|Q|\geq0 \]
To summarize, we showed the following two inequalities:
\begin{align*}
\cov(Q) & \geq \cov(\opt)-\const\cdot \cost(\opt)-\tau k\\
\cov(Q)-s\cdot \cost(Q) & \geq0
\end{align*}
Combining the two inequalities with coefficients $\const-1$ and $1$ gives
\begin{align*}
\cov(Q)-\cost(Q) & \geq\frac{\const-1}{\const}\left(\cov(\opt)-\const\cdot \cost(\opt)-\tau k\right)
\end{align*}
 \textbf{Setting $\const,\tau$. }We now put together the two cases and
set the two parameters $\const\geq1$ and $\tau$.

In case 1, we obtain a solution $Q$ with value
\[
\cov(Q)-\cost(Q)\geq \cov(Q)-\const\cdot \cost(Q)\geq\tau k
\]
where the first inequality is due to $c\geq0$ and $\const\geq1$, and
the second inequality is by our analysis above.

In case 2, we obtain a solution $Q$ with value
\[
\cov(Q)-\cost(Q)\geq\frac{\const-1}{\const}\left(\cov(\opt)-\const\cdot \cost(\opt)-\tau k\right)
\]
Thus overall we get a solution with value at least
\[
\min\left\{ \tau k,\frac{\const-1}{\const}\left(\cov(\opt)-\const\cdot \cost(\opt)-\tau k\right)\right\} 
\]
 We set $\tau$ to balance the two terms:
\begin{align*}
\tau k &=\frac{\const-1}{\const}\left(\cov(\opt)-\const\cdot \cost(\opt)-\tau k\right)\\
\Rightarrow
\tau k &=\frac{\const-1}{2\const-1}\left(\cov(\opt)-\const\cdot \cost(\opt)\right)
\end{align*}
We set $\const$ so that the coefficient of $\cost(\opt)$ becomes $1$:
\[
\frac{\const(\const-1)}{2\const-1}=1\Rightarrow \const^{2}-3\const+1=0
\]
The above equation has two solutions: $\const_{1}=\frac{1}{2}\left(3-\sqrt{5}\right)$ and $\const_{2}=\frac{1}{2}\left(3+\sqrt{5}\right)$. We want $\const\geq1$, so we pick the latter.
For this choice, the threshold $\tau$ and the objective value obtained are
\begin{align*}
\tau & =\frac{1}{k}\left(\frac{1}{2}\left(3-\sqrt{5}\right)\cov(\opt)-\cost(\opt)\right)\\
\cov(Q)-\cost(Q) & \geq\frac{1}{2}\left(3-\sqrt{5}\right)\cov(\opt)-\cost(\opt)
\end{align*}
\end{proof}

\spara{Guessing $\tau$ and the {\onlinekcostscaledgreedy} algorithm.}
Setting the threshold as suggested by the above theorem requires knowing $\hat{g}(\opt)$, where $\hat{g}(Q):=\frac{1}{2}(3-\sqrt{5})\cov(Q)-\cost(Q)$. To remove this assumption, we use the standard technique introduced by \cite{badanidiyuru2014streaming}, which we now sketch. The largest singleton value $v=\max_{e}\hat{g}(\{e\})$ gives us a $k$-approximation to $\hat{g}(\opt)$. Given this approximation, we guess a $1+\epsilon$ approximation to $\hat{g}(\opt)$ from a set of $O(\log k/\epsilon)$ values ranging from $v$ to $kv$. The final streaming algorithm is simply $O(\log k/\epsilon)$ copies of the basic algorithm running in parallel with different guesses. As new elements appear in the stream, the value $v=\max_{e}\hat{g}(\{e\})$ also increases over time and thus, existing copies of the basic algorithm with small guesses are dropped and new copies with higher guesses are added. Observe that when we introduce a new copy with a large guess, starting it from mid-stream has exactly the same outcome as if we started it from the beginning of the stream: all previous elements have marginal gain much smaller than the guess and smaller than the threshold so they would have been rejected anyway. We refer to \cite{badanidiyuru2014streaming} for the full details. We only lose $\epsilon$ in the approximation due to guessing and we use $O\left(k \log{k} / \epsilon \right)$ total space to store the $O(\log{k}/\epsilon)$ solutions. Thus, we have the following:

\begin{thm}
There is a streaming algorithm {\onlinekcostscaledgreedy} for the cardinality-constrained problem $\max_{|Q| \leq k} \cov(Q) - \cost(Q)$ that takes as input any $\epsilon > 0$ and it returns a solution $Q$ satisfying
\[ \cov(Q)-\cost(Q)\geq \left(\frac{1}{2}\left(3-\sqrt{5}\right) - \epsilon \right)\cov(\opt)-\cost(\opt) \]
The algorithm uses $O\left(k \log{k} / \epsilon \right)$ space.
\end{thm}

One can further improve the space usage by a $\log{k}$ factor using a modification of this technique due to Kazemi \etal~\cite{kazemi2019submodular} that approximates the value of the optimal solution by the maximum among the maximum singleton value and the best solution constructed so far, instead of only the maximum singleton value.   

\section{Algorithm for the Matroid-constrained Problem}
\label{app:matroid}

In this section, we extend the {\costscaledgreedy} algorithm from Section~\ref{sec:alg-cardinality} to the more general setting of a matroid constraint, i.e., the {\matroidconstrained} problem $\max_{Q \in \mathcal{I}} f(Q)-c(Q)$ where $\mathcal{I}$ is the collection of independent sets in a matroid $\mathcal{M} = (V, \mathcal{I})$. As before, the algorithm is the standard Greedy algorithm applied to the scaled objective $\tilde{g}(Q) = f(Q) - 2c(Q)$. We refer to this algorithm as {\matroidcostscaledgreedy} (Matroid Cost Scaled Greedy) and present its outline in Algorithm \ref{algo:matroid}.

\spara{Running time.}
The worst-case running time of {\matroidcostscaledgreedy} is $O(nk)$ evaluations of the functions $\cov$ and $\cost$ and $O(nk)$ matroid feasibility checks, where $n = |V|$ and $k$ is the rank of the matroid (the size of the largest independent set). The lazy evaluation technique that we described in Section~\ref{sec:alg-cardinality} can be used to speed up the matroid algorithm as well.

\begin{algorithm}[t]
\textbf{Input:} Ground set $V$, scaled objective $\tilde{g}(Q)=\cov(Q)-2\cost(Q)$, matroid $\mathcal{M} = (V, \mathcal{I})$. \\
\textbf{Output:} Solution $Q$.
\begin{algorithmic}[1]
\STATE $Q\gets\emptyset$, $N\gets V$
  \FOR{$i=1,\ldots,n$}
    \IF{$N=\emptyset$}
	    \STATE break
		\ENDIF
    \STATE $e_{i} = \arg\max_{e\in N} \tilde{g}(e|Q)$\label{line:condition}
    \IF{$\tilde{g}(e_i \vert Q) \leq 0$}
      \STATE break
    \ENDIF
	  \STATE $Q\gets Q\cup\{e_{i}\}$
		\STATE remove from $N$ every element $e$ s.t. $Q\cup \{e\}\notin \mathcal{I}$ \label{line:remove-infeasible}
  \ENDFOR
\RETURN{$Q$}
\end{algorithmic}
\caption{\label{algo:matroid}The {\matroidcostscaledgreedy} algorithm for the matroid-constrained problem {\matroidconstrained}.}
\end{algorithm}

\spara{Approximation guarantee.}
The following theorem states our approximation guarantee.
\begin{thm}
\label{thm:matroid}
Algorithm \ref{algo:matroid} returns a solution $Q\in\mathcal{I}$ satisfying $\cov(Q)-\cost(Q)\geq\frac{1}{2}\cov(\opt)-\cost(\opt)$.
\end{thm}
Note that the above guarantee matches the approximation of the standard Greedy algorithm for monotone functions (the special case when the costs are equal to $0$). There are simple examples of monotone submodular maximization with a partition matroid constraint for which Greedy only achieves a $\frac{1}{2}$ approximation. There are algorithms that achieve the optimal $1-\frac{1}{e}$ approximation for monotone submodular maximization with a matroid constraint, but these algorithms are generally very inefficient.

\subsection{Proof of Theorem~\ref{thm:matroid}}

In the remainder of this section, we prove Theorem~\ref{thm:matroid}. Note that the algorithm maintains the invariant that $Q \in \mathcal{I}$. Thus we focus on analyzing the objective value. The analysis is similar to the cardinality constraint. The main difference is in constructing the mapping between the solution $Q$ constructed by the algorithm and the optimal solution $\opt$. To this end, we use the following standard result for matroids due to Brualdi \cite{brualdi1969comments} (see also Chapter 39 of the textbook \cite{Schrijver2003}).

\begin{thm}[Brualdi's theorem]
\label{thm:mapping}
Let $I$ and $J$ be two independent sets in a matroid such that $|I|=|J|$. There is a bijection  $\pi\colon I\setminus J\to J\setminus I$ such that $(J\setminus\pi(e))\cup\{e\}$ is independent for every $e\in I\setminus J$.
\end{thm}

Let $Q^{(i)}$ be the solution $Q$ at the end of iteration $i$ of the algorithm. Let $Q$ and $N$ be the respective sets at the end of the algorithm. We partition $\opt$ into two sets $O_1$ and $O_2$, where $O_2 = \opt \cap N$ and $O_1 = \opt \setminus O_2$. Note that, if $N$ is non-empty, the algorithm terminated because $\tilde{g}(e \vert Q) \leq 0$ for all $e\in N$. 

We first show that $|Q| \geq |O_1|$. Suppose for contradiction that $|Q| < |O_1|$. By the augmentation property, there is an element $e \in O_1 \setminus Q$ such that $Q \cup \{e\} \in \mathcal{I}$. By the hereditary property, we have $Q^{(i)} \cup \{e\} \in \mathcal{I}$ for all $i \leq |Q|$. Thus $e$ could not have been removed from $N$ on line~\ref{line:remove-infeasible} and hence $e \in \opt \cap N = O_2$, contradicting the fact that $e \in O_1$.

Let $\ell = |O_1|$. We analyze $Q^{(\ell)}$ and show the following:

\begin{lem}
We have
\[ 2(\cov(Q^{(\ell)}) - \cost(Q^{(\ell)})) \geq \cov(Q^{(\ell)} \cup O_1) - 2\cost(O_1) \]
\end{lem}
\begin{proof}
Recall that we have $|Q^{(\ell)}| = |O_1| = \ell$. Moreover, $Q^{(\ell)}$ and $O_1$ are both independent. We apply Theorem~\ref{thm:mapping} to obtain a bijective mapping $\pi \colon O_1 \to Q^{(\ell)}$ such that $(Q^{(\ell)} \setminus \pi(e)) \cup \{e\} \in \mathcal{I}$ for all $e \in O_1 \setminus Q^{(\ell)}$. We augment $\pi$ by setting $\pi(e) = e$ for all $e \in O_1 \cap Q^{(\ell)}$, and obtain a bijection from $O_1$ to $Q^{(\ell)}$. We let $O_1^{(i)} = \pi^{-1}(Q^{(i)})$. 

Consider an iteration $i\leq \ell$ in which $Q^{(i)} \neq Q^{(i-1)}$ and thus $Q^{(i)} = Q^{(i-1)} \cup \{e_i\}$. Let $o_i = \pi^{-1}(e_i)$. If $o_i = e_i$, then clearly
\[ \tilde{g}(e_i \vert Q^{(i-1)}) = \tilde{g}(o_i \vert Q^{(i-1)}) \]
Suppose that $o_i \neq e_i$. We have $o_i \in O_1 \setminus Q^{(\ell)}$ and, by the choice of $\pi$, we have $(Q^{(\ell)} \setminus \{e_i\}) \cup \{ o_i \} \in \mathcal{I}$. Since $(Q^{(\ell)} \setminus \{e_i\}) \supseteq Q^{(j)}$ for all $j < i$, the hereditary property implies that $Q^{(j)} \cup \{o_i\} \in \mathcal{I}$ and thus $o_i$ could not have been removed from $N$ on line~\ref{line:remove-infeasible} in any iteration $j < i$. Thus $o_i \in N$ at the beginning of iteration $i$ and thus $o_i$ is a candidate for $e_i$. Therefore, by the choice of $e_i$, we have 
\[ \tilde{g}(e_i \vert Q^{(i-1)}) \geq \tilde{g}(o_i \vert Q^{(i-1)}) \]

Thus, for every iteration $i \leq \ell$ for which $Q^{(i)} = Q^{(i-1)} \cup \{e_i\}$, we have
\[
\cov(Q^{(i)})-\cov(Q^{(i-1)})-2\cost(e_{i})\geq \cov(Q^{(i-1)}\cup\{o_{i}\})-\cov(Q^{(i-1)})-2\cost(o_{i})
\]
where $o_i = \pi^{-1}(e_i)$.

We have
\[
\cov(Q^{(i-1)}\cup\{o_{i}\})-\cov(Q^{(i-1)})
= \cov(o_i \vert Q^{(i-1)})
\geq \cov(o_i \vert Q^{(\ell)} \cup O_1^{(i-1)})
= \cov(Q^{(\ell)} \cup O_1^{(i)})-\cov(Q^{(\ell)} \cup O_1^{(i-1)})
\]
where the inequality is by submodularity and the last equality is by $O^{(i)} = O^{(i-1)} \cup \{o_i\}$. 

Thus
\[
\cov(Q^{(i)})-\cov(Q^{(i-1)})-2\cost(e_{i})\geq \cov(Q^{(\ell)} \cup O_1^{(i)}) -\cov(Q^{(\ell)} \cup O_1^{(i-1)}) - 2\cost(o_{i})
\]
Summing up over all $i \leq \ell$ and using that $O_1^{(\ell)}=\pi^{-1}(Q^{(\ell)})=O_1$, we obtain
\begin{align*}
\cov(Q^{(\ell)})-2\cost(Q^{(\ell)}) & \geq \cov(Q^{(\ell)} \cup O_1) - \cov(Q^{(\ell)}) - 2\cost(O_1)\\
\Rightarrow 2(\cov(Q^{(\ell)})-\cost(Q^{(\ell)})) & \geq \cov(Q^{(\ell)} \cup O_1)-2\cost(O_1)
\end{align*}
\end{proof}

\begin{lem}
We have
\[ \cov(Q) \geq \cov(Q \cup O_2) - 2\cost(O_2) \]
\end{lem}
\begin{proof}
We may assume that $O_2 \neq \emptyset$, since otherwise the lemma is immediate. For every $o \in O_2$, we have $\tilde{g}(o \vert Q) \leq 0$ and thus
\[ 2\cost(o) \geq \cov(o \vert Q) \]
Summing up and using submodularity, we obtain
\[ 2\cost(O_2) \geq \sum_{o \in O_2} \cov(o \vert Q) \geq \cov(Q \cup O_2) - \cov(Q) \]
Rearranging, we obtain
\[ \cov(Q) \geq \cov(Q \cup O_2) - 2\cost(O_2) \]
\end{proof}
By combining the two lemmas and using submodularity, we obtain the following result.
\begin{lem}
We have
\[ \cov(Q^{(\ell)}) + \cov(Q) - 2\cost(Q^{(\ell)}) \geq \cov(\opt) - 2\cost(\opt)\]
\end{lem}
\begin{proof}
By combining the two lemmas above, we obtain
\[ 
2(\cov(Q^{(\ell)}) - \cost(Q^{(\ell)})) + \cov(Q) \geq \cov(Q^{(\ell)} \cup O_1) + \cov(Q \cup O_2) - 2\cost(\opt)
\]
Recall that $O_1$ and $O_2$ is a partition of $\opt$ and $Q^{(\ell)} \subseteq Q$. By submodularity, we have
\[ \cov(Q^{(\ell)} \cup O_1) + \cov(Q \cup O_2)
\geq \cov((Q^{(\ell)} \cup O_1) \cap (Q \cup O_2)) + \cov((Q^{(\ell)} \cup O_1) \cup (Q \cup O_2))
= \cov(Q^{(\ell)}) + \cov(Q \cup \opt)
\]
Therefore
\[ 
\cov(Q^{(\ell)}) - 2\cost(Q^{(\ell)}) + \cov(Q) \geq \cov(Q \cup \opt) - 2\cost(\opt)
\]
Since $\cov$ is monotone, we have $\cov(Q \cup \opt) \geq \cov(\opt)$, and thus
\[ 
\cov(Q^{(\ell)}) - 2\cost(Q^{(\ell)}) + \cov(Q) \geq \cov(\opt) - 2\cost(\opt)
\]
\end{proof}
Since the algorithm adds elements with positive marginal gain, we obtain the following result.
\begin{lem}
We have
\[ 2(\cov(Q)-\cost(Q)) \geq \cov(Q^{(\ell)}) + \cov(Q) - 2\cost(Q^{(\ell)}) \] 
\end{lem}
\begin{proof}
Since the algorithm adds elements with positive marginal gain, we have $\tilde{g}(Q^{(i)}) - \tilde{g}(Q^{(i-1)}) = \tilde{g}(e_i \vert Q^{i-1}) \geq 0$. Thus
\[ \tilde{g}(Q) - \tilde{g}(Q^{(\ell)}) = \sum_{i = \ell+1}^{|Q|} (\tilde{g}(Q^{(i)}) - \tilde{g}(Q^{i-1})) \geq 0 \]
Therefore
\[ \cov(Q) - 2\cost(Q) \geq \cov(Q^{(\ell)}) - 2\cost(Q^{(\ell)}) \]
and thus
\[ 2\cov(Q) - 2\cost(Q) \geq \cov(Q^{(\ell)}) + \cov(Q) - 2\cost(Q^{(\ell)}) \]
\end{proof}
Theorem~\ref{thm:matroid} now follows from the last two lemmas.

\section{Variants of {\costscaledgreedy} and {\matroidcostscaledgreedy}}
\label{sec:alg-variants}

In this section, we describe a small variant of {\costscaledgreedy} and its matroid counterpart {\matroidcostscaledgreedy}. As before, the algorithms run Greedy on the scaled objective $\tilde{g} = \cov - 2\cost$ but, instead of stopping as soon as the marginal gains with respect to $\tilde{g}$ become negative, they continue until the budget is reached and they return the prefix solution that has largest value with respect to the original objective $g$. Note that the solution of {\costscaledgreedy} (and {\matroidcostscaledgreedy}, respectively) is included as a prefix and thus the value of the solution constructed by these variants is at least as good as those of the original algorithms. Thus these variants retain the same approximation guarantee while at the same time being less conservative and potentially obtaining better objective. On the other hand, the original algorithms stop earlier and thus they may save on running time.

\begin{algorithm}[t]
\textbf{Input:} Ground set $V$, combined objective $g(Q)=\cov(Q)-\cost(Q)$, scaled objective $\tilde{g}(Q)=\cov(Q)-2\cost(Q)$, cardinality $k$. \\
\textbf{Output:} Solution $Q$.
\begin{algorithmic}[1]
\STATE $Q\gets\emptyset$
  \FOR{$i=1,\ldots,k$}
    \STATE $e_{i} = \arg\max_{e\in V \setminus Q}\tilde{g}(e|Q)$
	  \STATE $Q\gets Q\cup\{e_{i}\}$
  \ENDFOR
\RETURN{$\arg\max_{1\leq i \leq k} g(\{e_1, \dots, e_i\})$}
\end{algorithmic}
\caption{\label{algo:cardinality-variant} A variant of {\costscaledgreedy} algorithm for the cardinality-constrained problem {\constrained}.}
\end{algorithm}

\begin{algorithm}[t]
\textbf{Input:} Ground set $V$, combined objective $g(Q)=\cov(Q)-\cost(Q)$, scaled objective $\tilde{g}(Q)=\cov(Q)-2\cost(Q)$, matroid $\mathcal{M} = (V, \mathcal{I})$. \\
\textbf{Output:} Solution $Q$.
\begin{algorithmic}[1]
\STATE $Q\gets\emptyset$, $N\gets V$
  \FOR{$i=1,\ldots,n$}
    \IF{$N \setminus Q = \emptyset$}
	    \STATE break
		\ENDIF
    \STATE $e_{i} = \arg\max_{e\in N \setminus Q} \tilde{g}(e|Q)$
	  \STATE $Q\gets Q\cup\{e_{i}\}$
		\STATE remove from $N$ every element $e$ s.t. $Q\cup \{e\}\notin \mathcal{I}$
  \ENDFOR
\RETURN{$\arg\max_{1\leq i \leq |Q|} g(\{e_1, \dots, e_i\})$}
\end{algorithmic}
\caption{\label{algo:matroid-variant} A variant of {\matroidcostscaledgreedy} algorithm for the matroid-constrained problem {\matroidconstrained}.}
\end{algorithm}

These variants are simpler to analyze than the original algorithms. We include the simpler analysis for a cardinality constraint here.

\begin{thm}
Algorithm \ref{algo:cardinality-variant} returns a solution $Q$ of size at most $k$ satisfying $\cov(Q)-\cost(Q)\geq\frac{1}{2}\cov(\opt)-\cost(\opt)$.
\end{thm}
\begin{proof}
Let $Q^{(i)} = \{e_1, \dots, e_i\}$ be the elements selected in the first $i$ iterations. Let $\ell = |\opt|\leq k$. We will show that $g(Q^{(\ell)}) \geq \frac{1}{2} \cov(\opt) - \cost(\opt)$. Since the algorithm returns the prefix with maximum $g$ value, we obtain the same guarantee for the solution returned by the algorithm.

Let $\pi: \opt \to Q^{(\ell)}$ be an arbitrary bijective mapping that maps the elements of $\opt \cap Q^{(\ell)}$ to themselves. We now show that, for all $i \leq \ell$, we have
\[ \tilde{g}(e_i \vert Q^{(i-1)}) \geq \tilde{g}(\pi^{-1}(e_i) \vert Q^{(i-1)}) \]
The inequality trivially holds if $\pi^{-1}(e_i) = e_i$. Therefore we may assume that $\pi^{-1}(e_i) \neq e_i$ and thus $\pi^{-1}(e_i) \in \opt \setminus Q^{(\ell)} \subseteq V \setminus Q^{(i-1)}$. It follows that $\pi^{-1}(e_i)$ was a candidate for $e_i$. Thus, by the choice of $e_i$, we have
\[ \tilde{g}(e_i \vert Q^{(i-1)}) \geq \tilde{g}(\pi^{-1}(e_i) \vert Q^{(i-1)}) \]
Let $\opt^{(i)} = \pi^{-1}(Q^{(i)})$ for all $i \leq \ell$. By the above inequality and submodularity, for all $i \leq \ell$, we have
\[ \tilde{g}(e_i \vert Q^{(i-1)}) \geq \tilde{g}(\pi^{-1}(e_i) \vert Q^{(i-1)}) \geq \tilde{g}(\pi^{-1}(e_i) \vert Q^{(\ell)} \cup \opt^{(i-1)}) \]
Summing up over all $i \leq \ell$ and using that $\opt^{(i)} = \opt^{(i-1)} \cup \{\pi^{-1}(e_i)\}$ for all $i\leq \ell$ and $\opt^{(\ell)} = \pi^{-1}(Q^{(\ell)}) = \opt$, we obtain
\begin{align*}
\tilde{g}(Q^{(\ell)}) &\geq \tilde{g}(Q^{(\ell)}\cup \opt) - \tilde{g}(Q^{(\ell)})\\
\Rightarrow \tilde{g}(Q^{(\ell)}) &\geq \frac{1}{2} \tilde{g}(Q^{(\ell)}\cup \opt)
\end{align*}
Recall that $\tilde{g} = \cov - 2\cost$ and $g = \cov - \cost$. The above inequality gives
\begin{align*}
g(Q^{(\ell)}) 
&= \cov(Q^{(\ell)}) - \cost(Q^{(\ell)})\\
&\geq \frac{1}{2} \underbrace{\cov(Q^{(\ell)}\cup \opt)}_{\geq \cov(\opt)} - \underbrace{(\cost(Q^{(\ell)} \cup \opt) - \cost(Q^{(\ell)}))}_{=\cost(\opt \setminus Q^{(\ell)}) \leq \cost(\opt)}\\
&\geq \frac{1}{2} \cov(\opt) - \cost(\opt)
\end{align*}
where the last inequality is by monotonicity of $\cov$ and linearity and non-negativity of $\cov$.
\end{proof}

\end{document}